\documentclass[11pt]{article}

\usepackage[margin=0.95in]{geometry}
\usepackage{amsmath,amssymb,amsthm,mathtools}
\usepackage{booktabs,longtable,array,enumitem,pdflscape}
\usepackage{microtype}
\usepackage{etoolbox}
\usepackage{xcolor}
\usepackage[colorlinks=true,linkcolor=blue,citecolor=blue,urlcolor=blue]{hyperref}
\usepackage{authblk}

\makeatletter
\patchcmd{\l@section}{1.5em}{2.1em}{}{}
\renewcommand*\l@subsection{\@dottedtocline{2}{1.5em}{3.2em}}
\renewcommand*\l@subsubsection{\@dottedtocline{3}{4.7em}{3.8em}}
\makeatother

\newtheorem{theorem}{Theorem}[section]
\newtheorem{lemma}[theorem]{Lemma}

\newtheorem{corollary}[theorem]{Corollary}
\theoremstyle{definition}
\newtheorem{definition}[theorem]{Definition}
\newtheorem{example}[theorem]{Example}

\newcommand{\C}{\mathbb C}
\newcommand{\R}{\mathbb R}
\newcommand{\Q}{\mathbb Q}
\newcommand{\Z}{\mathbb Z}

\newcommand{\T}{\mathbb T}
\newcommand{\e}{\mathrm e}
\newcommand{\ii}{\mathrm i}
\newcommand{\wt}{\operatorname{wt}}
\newcommand{\supp}{\operatorname{supp}}
\newcommand{\rank}{\operatorname{rank}}
\newcommand{\diag}{\operatorname{diag}}
\newcommand{\aff}{\operatorname{aff}}

\newcommand{\spanR}{\operatorname{span}_{\R}}
\newcommand{\Span}{\operatorname{span}}
\newcommand{\tr}{\operatorname{tr}}

\newcommand{\ket}[1]{\lvert #1\rangle}
\newcommand{\bra}[1]{\langle #1\rvert}

\newcommand{\abs}[1]{\lvert #1\rvert}
\newcommand{\norm}[1]{\lVert #1\rVert}

\title{Finite Classification of Transversal Diagonal Gates\\ Beyond Stabilizer Codes}
\author[1]{Xi He}
\author[2]{Chao Zhang}
\author[2]{Shilin Huang}
\author[1]{Bei Zeng\thanks{Corresponding author: \href{mailto:bei.zeng@utdallas.edu}{bei.zeng@utdallas.edu}}}
\affil[1]{Department of Physics, The University of Texas at Dallas, Richardson, Texas 75080, USA}
\affil[2]{Department of Physics, The Hong Kong University of Science and Technology,\protect\\
Clear Water Bay, Kowloon, Hong Kong SAR, China}
\date{}
 
\begin{document}
\maketitle

\begin{abstract} 
We establish that transversal diagonal gates on exact qubit codes of distance at least two can be exhaustively classified within a finite search space. Although code-preserving physical gates can have continuous phases, any projective logical action induced by a transversal diagonal gate can always be realized by a physical gate of finite order. Crucially, we prove that this physical order has an upper bound determined solely by the code parameters. This order bound enables the subset-sum--linear-programming (SSLP) framework introduced in \href{https://arxiv.org/abs/2504.20847}{arXiv:2504.20847} to exhaustively enumerate all candidate discrete symmetries, after which the Knill--Laflamme equations determine the remaining continuous code parameters. Applying this method to the nontrivial diagonal sector of exact $((6,2,3))$ codes, SSLP yields $66$ candidate symmetry classes, of which exactly $15$ are realizable, supporting full projective diagonal logical cyclic groups of orders $2$, $3$, and $5$. In particular, for order $5$, the code is unique up to local unitaries and qubit permutations, equivalent to the code constructed in \href{https://arxiv.org/abs/2504.20847}{arXiv:2504.20847}. This code also attains the minimum length $n=6$ among exact $((n, 2, \ge 3))$ codes whose full projective transversal group cannot be conjugated into the projective Clifford group.

\end{abstract}

\section{Introduction}
\label{sec:introduction}

Fault-tolerant quantum computation requires both protection against physical
noise and operations that preserve that protection.  Transversal gates
address the second requirement by acting independently on the physical
subsystems of a code block, preventing a single fault from spreading within
the block \cite{Shor1996FaultTolerant,Gottesman1997Stabilizer}.  However,
transversal gates alone cannot provide universal quantum computation on a
fixed error-detecting code: their combinations cannot approximate every
logical unitary.  This obstruction holds for additive quantum codes
\cite{ZengCrossChuang2011} and, more generally, for arbitrary exact codes
that detect errors on each individual physical subsystem
\cite{EastinKnill2009}.  Additional structure yields sharper constraints,
including dimensional restrictions on locality-preserving gates of
topological stabilizer codes \cite{BravyiKoenig2013,PastawskiYoshida2015}
and Clifford-hierarchy restrictions on transversal diagonal gates of qubit
stabilizer codes \cite{AndersonJochymOConnor2016}.  These results leave a
complementary classification problem: \emph{which nontrivial transversal
logical symmetries can exact quantum codes support, and which codes
realize them?}

For stabilizer codes, the classification problem benefits from a finite
algebraic description in terms of Pauli operators and their commutation
relations \cite{Gottesman1997Stabilizer,CalderbankRainsShorSloane1998}.
This structure supports both classification theorems and exhaustive
catalogues, including the enumeration of binary stabilizer codes through
nine qubits up to local Clifford transformations and qubit permutations
\cite{CrossVandeth2025}. For CSS and stabilizer codes, algebraic conditions
and algorithms also characterize diagonal code-preserving gates and their
induced logical actions
\cite{HuLiangCalderbank2022,WebsterQuintavalleBartlett2023}. General quantum
codes require a broader description: the code subspace can vary
continuously within the physical Hilbert space. The Knill--Laflamme
conditions impose nonlinear constraints on this choice, while an unknown
transversal symmetry introduces additional continuous parameters through
its local unitary factors. An exhaustive classification must therefore
control both the continuous geometry of the codes and the possible
physical gates that preserve them.

Nonadditive codes make this classification problem particularly rich.
Early constructions exceeded the dimensions attainable by stabilizer codes
at fixed length and distance
\cite{RainsHardinShorSloane1997,YuChenLaiOh2008}, while distance-two bounds
and constructions revealed further departures from additive behavior
\cite{Rains1999DistanceTwo}. Graphical and codeword-stabilized descriptions
\cite{YuChenOh2007Graphical,CrossSmithSmolinZeng2009,ChuangCrossSmithSmolinZeng2009},
constructions from nonlinear classical codes \cite{GrasslRoetteler2008},
and permutation-invariant codes \cite{PollatsekRuskai2004,Ouyang2014}
have provided systematic ways to explore this larger space. The XP
formalism extends stabilizer descriptions using finite-order phase
rotations and supplies algorithms for determining codespaces and logical
operators \cite{WebsterBrownBartlett2022}. These developments also reveal
new possibilities for transversal gates: nonadditive codes realize
binary-icosahedral transversal symmetry and generalized phase gates
\cite{KubischtaTeixeira2023,KubischtaTeixeira2025}, while constructions based
on twisted unitary \(t\)-groups connect prescribed transversal symmetries
with quantum error correction \cite{KubischtaTeixeira2024}. At the same
time, continuous families of exact \(((6,2,3))\) and \(((7,2,3))\) codes
show that even small code parameters support substantial variation beyond
a finite stabilizer catalogue \cite{DuZhangPoonZeng2026}. An exhaustive
framework must accommodate these continuous families while determining
which transversal symmetries they can support.

The subset-sum--linear-programming (SSLP) framework introduced in
Ref.~\cite{ZhangWuHuangZeng2025} suggests a way to separate the discrete
symmetry data from the continuous code parameters. Once a finite-order
diagonal product gate is fixed, its eigenspaces partition the computational
basis into modular subset-sum sectors. Diagonal error detection supplies
linear moment constraints on probability distributions over these sectors,
and the remaining Knill--Laflamme equations form an exact completion
problem for the complex amplitudes. This framework produced a
\(((6,2,3))\) code admitting a cyclic projective logical group \(C_5\)
of order five, realized by diagonal transversal gates, together with
\(((7,2,3))\) codes admitting transversal \(T\) and \(\sqrt T\) gates
\cite{ZhangWuHuangZeng2025}. These constructions leave open the question of
completeness. SSLP begins with a finite-order physical gate, but a
code-preserving diagonal product unitary can have infinite physical order
even when its projective logical action has finite order. Even the least
physical order \(q_{\min}\) realizing a logical action can exceed its
projective order \(m\). An exhaustive classification using SSLP therefore
requires two ingredients: every relevant logical action must have a
finite-order physical representative, and the required physical order must
be bounded uniformly at fixed code parameters.

Our main result, established in Section~\ref{sec:sslp-classification}, is an
\emph{exhaustive SSLP normal form} for exact qubit codes of distance at least
two admitting a nontrivial transversal diagonal logical action.
Every such code admits an SSLP description: the eigenspaces of the chosen
logical action are supported on labeled modular subset-sum sectors, and
the code amplitudes satisfy the corresponding Knill--Laflamme equations.
Conversely, every exact completion of these data defines a code of
at least the target distance with the prescribed logical action.
This equivalence establishes the completeness of the SSLP framework,
without any stabilizer or additivity assumption, and includes
continuous families of exact codes. 

We further prove that this normal form requires only finitely many
choices of symmetry data at fixed code parameters. In a fixed product
basis, the projective diagonal logical group is finite abelian, and
every element has a finite-order physical representative, even when
the original physical gate has infinite order. We bound the required
physical order uniformly in the code parameters and obtain a stronger
bound for one logical qubit at distance at least three. Together,
these results reduce exhaustive classification to a finite enumeration
of discrete symmetry data followed by an exact algebraic completion
problem. These structural results were derived by the human
authors and subsequently verified in Lean, with large language models
(LLMs) writing the proof code under human guidance.

As an application of this framework, Section~\ref{six:sec:scope} classifies
the transversal diagonal symmetry data of exact \(((6,2,3))\) codes. Within
the finite physical-order cutoff, the subset-sum (SS) and
linear-programming (LP) stages yield \(66\) candidate classes of labeled
subset-sum data, modulo local bit flips, qubit permutations, and interchange
of the two logical labels, with the physical and logical order labels
retained. At the full Knill--Laflamme (KL) stage, we reuse the TeXRA-based
multi-agent workflow and independent Lean verification agent developed in
Ref.~\cite{HeLuZeng2026} to resolve the remaining coupled polynomial
constraints on code amplitudes and phases. The agents derive exact no-go
proofs excluding \(51\) classes and explicit solutions establishing
realizability for the remaining \(15\). The human authors audited these
results. These classes organize the allowed symmetry data: a realizable
class can contain continuous families of codes, and the same code can
occur in several classes. 

Classification of their code realizations up to
local unitaries is a further problem. Across all realizations, the possible
nontrivial full projective diagonal logical groups are precisely \(C_2\),
\(C_3\), and \(C_5\), where \(C_m\) denotes the cyclic group of order \(m\).
We determine the complete sets of attainable weight
enumerators for thirteen of the fifteen realizable classes. The order-five
sector admits the stronger conclusion of a complete code-level
classification: every exact \(((6,2,3))\) code admitting such a diagonal
logical action is equivalent, up to local unitaries and qubit permutations,
to the order-five code of Ref.~\cite{ZhangWuHuangZeng2025}. This code also
attains the minimum length of six for an exact code encoding one logical
qubit at distance at least three and supporting a single-block transversal
non-Clifford gate. The minimum-length statement is understood in a
logical-basis-independent sense: the full projective transversal logical
group cannot be conjugated into the projective Clifford group. The
Supplementary Materials and supporting code and data are available in the
companion GitHub repository~\cite{SSLPClassificationSupplement}.

\section{Preliminaries}
\label{sec:preliminaries}

To formulate the classification problem precisely, we first introduce the
notions of exact error detection, projective logical action, and SSLP data.
The Knill--Laflamme conditions govern exact error detection, while the
projective viewpoint distinguishes a logical action from its physical
representatives.  SSLP then separates the search into discrete subset-sum
data and continuous exact-completion constraints.  Throughout, we work with
exact, finite-dimensional qubit codes, without assuming additivity or
stabilizer structure unless explicitly stated.

\subsection{Quantum codes and exact error correction}
\label{sec:qec-preliminaries}

We begin with the form of the Knill--Laflamme conditions needed later.  In
particular, separating their diagonal and off-diagonal matrix elements will
allow us to distinguish the linear moment constraints used by SSLP from the
remaining exact completion equations.

The physical Hilbert space of \(n\) qubits is
\(
  \mathcal H_n=(\C^2)^{\otimes n}
\).
Let
\(
  \mathcal P_n=\{I,X,Y,Z\}^{\otimes n}
\)
denote the phase-free Pauli basis, and let \(\wt(E)\) be the number of
nonidentity tensor factors of \(E\in\mathcal P_n\).  A quantum code is a
subspace
\begin{equation}
  \mathcal Q\subseteq\mathcal H_n,
  \qquad
  \dim\mathcal Q=K,
\end{equation}
with orthogonal projector \(P\).  We write \(((n,K,d))\) for a general
qubit code.  Unlike the stabilizer notation \([[n,k,d]]\), this notation does
not require \(K=2^k\).

For a finite error set \(\{E_\alpha\}\), exact correction is possible if and
only if there is a Hermitian matrix \((c_{\alpha\beta})\) such that
\begin{equation}
  P E_\alpha^\dagger E_\beta P
  =c_{\alpha\beta}P
  \qquad\text{for all }\alpha,\beta.
  \label{eq:general-KL}
\end{equation}
These are the Knill--Laflamme conditions
\cite{KnillLaflamme1997}.  Equivalently, the code has distance at least
\(d_0\) precisely when
\begin{equation}
  PEP=\lambda_E P
  \label{eq:KL}
\end{equation}
for every \(E\in\mathcal P_n\) with \(\wt(E)<d_0\).  Thus a distance-
\(d\) code detects all errors on fewer than \(d\) qubits and corrects all
errors on at most \(\lfloor(d-1)/2\rfloor\) qubits.  We use ``exact'' to
emphasize that Eqs.~\eqref{eq:general-KL}--\eqref{eq:KL} are imposed as
equalities, rather than approximately.

In any orthonormal logical basis \(\{\ket{r_L}\}_{r=0}^{K-1}\),
Eq.~\eqref{eq:KL} is equivalent to
\begin{equation}
  \bra{r_L}E\ket{s_L}=\lambda_E\delta_{rs},
  \qquad 0\le r,s<K,\quad \wt(E)<d_0.
  \label{eq:logical-basis-KL}
\end{equation}
Thus every such error has equal diagonal matrix elements on all logical
basis states and vanishing off-diagonal matrix elements.  This separation is
the bridge to SSLP: the diagonal requirements become moment equations for
squared amplitudes, while the off-diagonal requirements remain part of the
exact complex-amplitude completion.

The quantum Singleton bound \cite{KnillLaflamme1997} gives
\begin{equation}
  K\le 2^{n-2(d-1)}.
  \label{eq:quantum-Singleton}
\end{equation}
In particular, a six-qubit code encoding one qubit has \(d\le3\).  We will
therefore certify exact \(((6,2,3))\) codes by imposing Eq.~\eqref{eq:KL} for
all Pauli operators of weight at most two.

To compare the exact codes obtained after completion, we use invariants that
do not depend on a chosen logical basis.  The Shor--Laflamme weight
enumerators of a rank-$K$ code projector $P$ are
\cite{ShorLaflamme1997,Rains1998WeightEnumerators}
\begin{align}
 A_w(P)
 &:=
 \frac1{K^2}
 \sum_{\substack{E\in\mathcal P_n\\\wt(E)=w}}
 \abs{\tr(EP)}^2,
 \label{eq:A-enumerator}\\
 B_w(P)
 &:=
 \frac1K
 \sum_{\substack{E\in\mathcal P_n\\\wt(E)=w}}
 \tr(EPE^\dagger P).
 \label{eq:B-enumerator}
\end{align}
Their homogeneous enumerator polynomials are
\begin{equation}
 A_P(x,y):=\sum_{w=0}^n A_w(P)x^{n-w}y^w,
 \qquad
 B_P(x,y):=\sum_{w=0}^n B_w(P)x^{n-w}y^w,
 \label{eq:enumerator-polynomials}
\end{equation}
and the quantum MacWilliams identity reads
\begin{equation}
 B_P(x,y)=\frac{K}{2^n}A_P(x+3y,x-y).
 \label{eq:quantum-MacWilliams}
\end{equation}
For a distance-$d$ code, the Knill--Laflamme conditions imply
\begin{equation}
 A_w(P)=B_w(P),
 \qquad 0\le w<d.
 \label{eq:enumerator-distance}
\end{equation}

The finite search and the final code classification use different notions
of equivalence, so we distinguish them at the outset.  Two code projectors
$P$ and $P'$ are \emph{local-unitary equivalent}, with
physical-qubit permutations allowed, when
\begin{equation}
 P'
 =
 \Pi_\sigma
 \left(\bigotimes_{j=1}^nU_j\right)
 P
 \left(\bigotimes_{j=1}^nU_j^\dagger\right)
 \Pi_\sigma^\dagger
 \label{eq:LU-equivalence}
\end{equation}
for one-qubit unitaries $U_j$ and a physical-qubit permutation
$\Pi_\sigma$.  We also consider the more restrictive equivalence that
allows only single-qubit bit flips and diagonal phase gates, together with
qubit permutations.  In that case $U_j=\Lambda_jX^{b_j}$, where
$\Lambda_j$ is diagonal unitary and $b_j\in\{0,1\}$.  This is often called
\emph{local-monomial equivalence}.  It implies local-unitary equivalence.
The Shor--Laflamme enumerators are invariant under both relations, but
equality of enumerators does not in general imply either equivalence.

We label the computational basis by bit strings
\(x=(x_1,\ldots,x_n)\in\{0,1\}^n\).  For
\(
  \ket\psi=\sum_x c_x\ket{x}
\), write
\(
  \supp(\psi)=\{x:c_x\ne0\}
\).  The distinction between the full set of allowed bit strings and the
actual support will be important in the SSLP completion step.
On bit strings, local $X$ operations and physical-qubit permutations act as
$x\mapsto\sigma(x)\oplus b$, where $b\in\{0,1\}^n$ and $\oplus$ denotes
bitwise addition with $1+1=0$.  These bit flips and coordinate permutations
are used to identify equivalent subset-sum data in
Section~\ref{six:sec:scope}; diagonal local phases leave the supports unchanged.
They will also provide the canonical normalizations used to choose
representatives of SSLP symmetry data.

\subsection{Transversal gates and restrictions on universality}
\label{sec:transversal-preliminaries}
\label{sec:transversal-no-go}

Our classification concerns logical symmetries, whereas SSLP is formulated
in terms of physical product gates.  We must therefore distinguish the
logical action of a transversal gate from the order and form of its physical
representatives.

A single-block transversal unitary is a product
\begin{equation}
  U=\bigotimes_{j=1}^n U_j,
  \qquad U_j\in U(2).
  \label{eq:transversal-product}
\end{equation}
More generally, a multiblock transversal operation acts independently on
corresponding coordinates of the blocks.  This factorization prevents one
fault from spreading to several coordinates within the same code block and
is one of the basic mechanisms of fault-tolerant quantum computation
\cite{Shor1996FaultTolerant}.  The present paper concerns the single-block
form \eqref{eq:transversal-product}.

A physical unitary \(U\) preserves the code when
\begin{equation}
  UPU^\dagger=P.
  \label{eq:code-preserving-unitary}
\end{equation}
It then induces a logical unitary \(U|_{\mathcal Q}\).  An overall phase
has no physical effect, so we identify logical unitaries that differ only
by such a phase.  Write $U(1)=\{z\in\C:|z|=1\}$ for the group of scalar
phases.  We write \(\overline U\) for the equivalence class of
\(U|_{\mathcal Q}\) under this identification.
Its \emph{projective logical order} is the least positive number of
repetitions that returns the logical operation to an overall phase times
the identity, if such a number exists.  In contrast, the physical order
requires the physical operator itself to return to the identity.  A
\emph{physical representative} of \(\overline U\) is a code-preserving
physical gate that induces that logical action up to an overall phase.
The action is \emph{projectively nontrivial} when it does more than multiply
every code state by the same phase.

After fixing one local product basis, a transversal diagonal unitary can be
written, up to a scalar phase, as
\begin{equation}
  D(\theta):=\bigotimes_{j=1}^n
  \diag(1,\e^{2\pi\ii\theta_j}),
  \qquad \theta\in\R^n.
  \label{eq:Dtheta}
\end{equation}
Each parameter \(\theta_j\) measures a phase in full turns: its phase angle
is \(2\pi\theta_j\).  Adding any whole number of turns leaves the gate
unchanged, so we can choose \(0\leq\theta_j<1\).
The word ``diagonal'' is basis dependent.  Local changes of basis may
diagonalize one commuting family of product gates, but need not
simultaneously diagonalize noncommuting transversal generators.

In the chosen product basis, the \emph{projective diagonal logical group} is
\begin{equation}
  \mathcal G_{\mathrm{diag}}(P)
  :=
  \left\{
  \overline{D(\theta)}:
  \theta\in\R^n,\quad
  D(\theta)PD(\theta)^\dagger=P
  \right\}.
  \label{eq:Gdiag}
\end{equation}
This group is abelian because diagonal physical gates commute.  For one
logical qubit (\(K=2\)), their logical restrictions share an eigenbasis, and
each projective logical action is determined by the ratio of its two
eigenvalues.  Consequently, if the group is finite, it is cyclic:
\[
  \mathcal G_{\mathrm{diag}}(P)\cong C_m,
  \qquad m=\abs{\mathcal G_{\mathrm{diag}}(P)}.
\]

A logical gate set is \emph{universal} if its products approximate arbitrary
logical unitaries up to global phase.  Transversal gates alone cannot
provide universality for stabilizer and subsystem stabilizer codes under
the conditions of Refs.~\cite{ZengCrossChuang2011,ChenChungCrossZengChuang2008},
including operations across multiple code blocks.  More generally, any
nontrivial finite-dimensional code that exactly detects every error on each
individual physical subsystem has a finite projective logical group of
code-preserving product unitaries \cite{EastinKnill2009}.  Individual
non-Clifford transversal gates remain possible.

Additional assumptions yield stronger restrictions.  For two-dimensional
topological stabilizer codes of sufficiently large distance, constant-depth
geometrically local logical circuits act within the Clifford group
\cite{BravyiKoenig2013}.  For qubit stabilizer codes, diagonal logical gates
implemented by transversal diagonal unitaries belong to the Clifford
hierarchy \cite{AndersonJochymOConnor2016}.

These results strongly constrain the possible logical actions, but they do
not yet provide a finite-order physical representative or a uniform bound
on its order.  Without those two ingredients, an exhaustive SSLP search
would still have an unbounded continuous input.  Closing precisely this gap
is the purpose of Section~\ref{sec:sslp-classification}.

\subsection{The SSLP framework}
\label{sec:sslp-preliminaries}

Once a finite-order physical representative is available, subset-sum linear
programming (SSLP), proposed in Ref.~\cite{ZhangWuHuangZeng2025}, turns the
problem for transversal diagonal gates into three layers: modular subset-sum solution
sets (SS), linear programming for diagonal moments (LP), and exact
completion of the remaining Knill--Laflamme equations (KL).  The key point
is that the first layer is discrete, the second is convex, and only the
third retains the full complex-amplitude problem.  By \emph{exact
completion}, we mean choosing complex amplitudes and phases that satisfy all
of the error-correction equations.
We describe these layers for one logical qubit, emphasizing the
distance-three conditions used in the six-qubit application.
Section~\ref{sec:sslp-normal-form} gives the general-rank formulation.

The construction starts with a prescribed diagonal product gate whose phase
angles are rational multiples of $2\pi$.  Such a gate has finite order and
can be written as
\begin{equation}
 D(a;q):=\bigotimes_{j=1}^n\diag(1,\omega_q^{a_j}),
 \qquad
 \omega_q:=\e^{2\pi\ii/q},
 \qquad
 a\in\Z_q^n.
 \label{eq:Daq}
\end{equation}
Here \(q\ge2\) is an integer and \(\Z_q:=\Z/q\Z\).  This is the special
case \(\theta=a/q\) of Eq.~\eqref{eq:Dtheta}; the existence of a finite
representative for an arbitrary projectively nontrivial code-preserving
diagonal action under detection of every \(Z_j\) is proved in
Section~\ref{sec:finite-projective-group}.
The representation is in \emph{reduced form} when $q,a_1,\ldots,a_n$ have
no common factor greater than one, that is,
\begin{equation}
 \gcd(q,a_1,\ldots,a_n)=1.
 \label{eq:primitive}
\end{equation}
The displayed physical product gate then has exact order \(q\): applying it
\(q\) times gives the identity, and no smaller positive number does so.
Reduced form does not mean that \(q\) is the smallest physical order among
all representatives of the same logical action.  If its induced
projective logical action has order \(m\), then \(m\mid q\), but in general
\(m\ne q\).  For a projective logical element \(\overline U\) admitting a
finite-order physical representative, define
\begin{equation}
 q_{\min}(\overline U)
 :=\min\{q:\ \overline U\text{ has a representative in reduced form }D(a;q)\}.
 \label{eq:qmin}
\end{equation}
The minimum is taken for the fixed code and the chosen local product basis.

The diagonal action on \(\ket{x}\), for \(x\in\{0,1\}^n\), is governed by
the modular subset sum
\begin{equation}
 A_a(x):=a\cdot x=\sum_{j=1}^n a_jx_j\pmod q
 \label{eq:subset-sum}
\end{equation}
and, for each residue \(b\), the corresponding subset-sum solution set
\begin{equation}
 S_b(a;q):=\{x\in\{0,1\}^n:A_a(x)\equiv b\pmod q\}.
 \label{eq:fiber}
\end{equation}
Indeed, \(D(a;q)\ket{x}=\omega_q^{A_a(x)}\ket{x}\).  Suppose this gate
preserves a rank-two code and induces a projective logical action of order
\(m>1\).  Its two logical eigenvectors have supports contained in distinct
subset-sum solution sets.
Up to local \(X\) conjugation, physical-qubit permutation, a common scalar,
logical-state interchange, and replacing the gate by its \(u\)th power
with \(\gcd(u,q)=1\), an orthonormal logical eigenbasis has the form
\begin{equation}
  \ket{0_L}=\sum_{x\in S_0(a;q)}c_x\ket{x},
  \qquad
  \ket{1_L}=\sum_{y\in S_{q/m}(a;q)}d_y\ket{y}.
  \label{eq:prelim-two-fiber}
\end{equation}
Here the normalized data are again denoted by \((q,m,a)\).  The local
\(X\) conjugation sends a supported string of the first logical vector to
\(0^n\), placing that vector in \(S_0\) after the common scalar is removed.
Choosing a power relatively prime to \(q\) normalizes the relative residue
to \(q/m\) while preserving both orders.  These steps are justified in
Corollary~\ref{cor:rank-two-sslp}.  Zero coefficients are allowed throughout.

The squared amplitudes \(p_0(x):=|c_x|^2\) and \(p_1(y):=|d_y|^2\)
are probability distributions on the two solution sets:
\begin{equation}
  p_0(x),p_1(y)\ge0,
  \qquad
  \sum_{x\in S_0}p_0(x)=\sum_{y\in S_{q/m}}p_1(y)=1.
  \label{eq:prelim-LP-normalization}
\end{equation}
For a diagonal Pauli operator \(Z_J:=\prod_{j\in J}Z_j\),
\begin{equation}
  Z_J\ket{x}=\prod_{j\in J}(1-2x_j)\ket{x}.
  \label{eq:prelim-diagonal-Pauli}
\end{equation}
Because the two solution sets are disjoint, the off-diagonal logical matrix
elements of every \(Z_J\) vanish automatically.  Equality of the diagonal
matrix elements in Eq.~\eqref{eq:logical-basis-KL} for all \(|J|<d_0\)
is equivalent to matching the averages of all products of fewer than
\(d_0\) distinct bits.  These averages are the \emph{binary moments};
their degree is the number of bits in the product.  The equivalence follows
because the substitutions \(z_j=1-2x_j\) and \(x_j=(1-z_j)/2\)
express either family of functions in terms of the other.

The first moments are the averages of individual bits.  The second moments
are the averages of products $x_ix_j$; since $x_i^2=x_i$, only distinct pairs
add new information.  In particular, at distance at least three, the
conditions for \(Z_i\) and \(Z_iZ_j\) give the moment equations
\begin{equation}
  \sum_{x\in S_0}p_0(x)\phi_2(x)
  =
  \sum_{y\in S_{q/m}}p_1(y)\phi_2(y),
  \qquad
  \phi_2(x):=(1,x_1,\ldots,x_n,(x_ix_j)_{i<j}).
  \label{eq:prelim-moment-LP}
\end{equation}
Equations~\eqref{eq:prelim-LP-normalization}
and~\eqref{eq:prelim-moment-LP} ask whether nonnegative weights satisfying
these linear equations exist.  This test depends only on the two full
subset-sum solution sets.  Equivalently, the sets of
all weighted averages of \(\phi_2(S_0)\) and \(\phi_2(S_{q/m})\) must
intersect, where the weights are nonnegative and sum to one.  These sets
of weighted averages are called convex hulls.  This is the moment
geometry retained by an SSLP moment candidate.  At distance two only the
first moments enter the LP; higher target distances add higher-degree moments.

For general rank \(K\), distinct logical eigenvalues correspond to distinct
subset-sum solution sets.
Normalized diagonals of the corresponding eigenspace projectors provide
probability distributions whose moments must agree across these solution sets.
When an eigenspace has dimension greater than one, these averaged moment
conditions do not capture all logical matrix constraints within that
eigenspace; those remain part of the full KL completion described in
Section~\ref{sec:sslp-normal-form}.

For rank two, the complete workflow is therefore
\begin{equation}
  (q,m,a)
  \xrightarrow{\mathrm{SS}}
  (S_0,S_{q/m})
  \xrightarrow{\mathrm{LP}}
  \text{feasible moment data}
  \xrightarrow{\mathrm{KL}}
  \text{exact code realizations}.
  \label{eq:sslp-three-layers}
\end{equation}
The SS layer is discrete; the LP layer tests only diagonal moment equations;
and the KL layer restores the complex amplitudes and imposes every remaining
Pauli equation, including the non-diagonal \(X\)- and \(Y\)-type
constraints.  For each such error, both requirements of
Eq.~\eqref{eq:logical-basis-KL} must hold.  Passing the LP test is therefore
necessary but not sufficient for a
physical code.  In
Section~\ref{six:sec:scope}, we group outputs of the first two layers that
are related by bit flips, coordinate permutations, and interchange of
the two logical labels.  Each group is an \emph{SSLP candidate class}.
It is \emph{realizable} only when some amplitudes and phases satisfy all
remaining Knill--Laflamme equations.

This decomposition makes SSLP a finite search once the physical order
\(q\) is bounded.  What remains unclear is whether starting from a
finite-order physical gate loses any diagonal logical symmetries, and
whether the necessary values of \(q\) can be bounded before the code is
known.  Section~\ref{sec:sslp-classification} answers both questions.

\section{SSLP normal form and physical-order bounds}
\label{sec:sslp-classification}

Section~\ref{sec:sslp-preliminaries} left two questions unresolved.  Does
every nontrivial diagonal logical action admit a finite-order physical
representative, and can the required physical order be bounded before the
code is known?  We answer both questions here, thereby turning SSLP from a
construction method into an exhaustive finite classification framework.

Our argument proceeds in three steps.  We first describe all
code-preserving diagonal phases by an integer matrix.  The homogeneous real
directions of this system act only by logical scalars, while every remaining
projective class has a representative with a common rational denominator.
The logical eigenspaces of such a representative then give the modular
subset-sum data of the SSLP normal form.  Finally,
Section~\ref{sec:q-bounds} bounds the necessary physical order \(q\)
uniformly in terms of the code parameters, making the outer enumeration
finite.  Section~\ref{six:sec:scope} applies the resulting framework to
six-qubit codes.  Throughout, ``diagonal'' refers to the fixed local product
basis chosen at the outset.

\subsection{Finite diagonal logical groups and physical representatives}
\label{sec:finite-projective-group}

We now establish finiteness of the projective diagonal logical group
defined in Eq.~\eqref{eq:Gdiag} and the existence of finite-order physical
representatives.

The error-detection input is particularly modest: we use only the
single-site $Z_j$ conditions.  Set
\begin{equation}
 N_j:=\frac{I-Z_j}{2}.
 \label{eq:number-operator}
\end{equation}
On a computational-basis vector $\ket{x}$, $N_j$ has eigenvalue $x_j$.
Detection of $Z_j$ is equivalent to
\begin{equation}
 PN_jP=\mu_jP,
 \qquad
 \mu_j=\frac{1-\lambda_j}{2}.
 \label{eq:Nj-detection}
\end{equation}

The phase parameters in Eq.~\eqref{eq:Dtheta} form the $n$-torus, which
identifies phase vectors differing by integers in each coordinate:
\[
 \begin{aligned}
  \T^n:=\R^n/\Z^n
  &\cong (U(1))^n \\
  &=\left\{(\e^{2\pi\ii\theta_1},\ldots,\e^{2\pi\ii\theta_n}):
    \theta_j\in\R\right\}.
 \end{aligned}
\]
The code-preserving phase vectors form the subgroup
\begin{equation}
 \mathcal T_P
 :=\{\theta\in\T^n:D(\theta)PD(\theta)^\dagger=P\}.
 \label{eq:physical-diagonal-stabilizer}
\end{equation}
The preservation condition will become a finite system of integer
constraints.  The following elementary lemma is the step that replaces an
arbitrary real solution by a rational representative without changing its
homogeneous component.

\begin{lemma}[Uniform rational section]
\label{lem:uniform-rational-section}
Let $A\in\Z^{s\times n}$.  There is a positive integer $Q_A$ such that, for
every $t\in\R^n$ satisfying $At\in\Z^s$, there is an $a\in\Z^n$ with
\begin{equation}
 A\left(t-\frac{a}{Q_A}\right)=0.
 \label{eq:uniform-rational-section}
\end{equation}
\end{lemma}

\begin{proof}
Regard $A$ as a rational matrix, with rational matrices acting on real
vectors by scalar extension, and set $B:=A^{\mathsf T}A\in\Q^{n\times n}$.
Choose a rational generalized inverse $G$ of $B$, so that
\begin{equation}
 BGB=B.
 \label{eq:rational-generalized-inverse}
\end{equation}
Such a $G$ exists: choose a right inverse of $B$ on $\operatorname{im}B$
and extend it arbitrarily to a complementary subspace.  Choose $Q_A>0$ so
that $M:=Q_AG$ has integer entries.

Now let $b:=At\in\Z^s$, set $d:=A^{\mathsf T}b\in\Z^n$, and define
$a:=Md\in\Z^n$.  Since $d=A^{\mathsf T}At=Bt$, Eq.~\eqref{eq:rational-generalized-inverse}
gives
\[
 B\frac{a}{Q_A}=BGd=BGBt=Bt.
\]
Hence $B(t-a/Q_A)=0$.  Over $\R$ one has
$\ker(A^{\mathsf T}A)=\ker A$, because
\[
 v^{\mathsf T}A^{\mathsf T}Av=\norm{Av}^2.
\]
Therefore $A(t-a/Q_A)=0$.
\end{proof}

We can now answer the first question posed at the beginning of the section.
The theorem concerns all logical operations induced by code-preserving
diagonal product gates, with overall phases ignored.  Its proof combines
the single-site detection condition in Eq.~\eqref{eq:Nj-detection} with the
rational replacement supplied by
Lemma~\ref{lem:uniform-rational-section}.  A group is \emph{abelian} when all
its elements commute.

\begin{theorem}[Finite diagonal logical group and physical representatives]
\label{thm:finite-projective-group}
Fix a local product basis, and let $Z_j$ denote Pauli $Z$ on the $j$th qubit
in that basis.  Suppose a rank-$K$ code projector $P$ satisfies
\begin{equation}
 PZ_jP=\lambda_jP,
 \qquad j=1,\ldots,n.
 \label{eq:all-Z-detection}
\end{equation}
Then $\mathcal G_{\mathrm{diag}}(P)$ is a finite abelian group.
Moreover, every nonidentity element of projective logical order $m$ has a
physical representative in reduced form $D(a;q)$ of exact order $q$, with
$m\mid q$.
\end{theorem}

\begin{proof}
For $x\in\{0,1\}^n$ and $t\in\R^n$, write
\[
 s_t(x):=x\cdot t=\sum_{j=1}^n x_jt_j.
\]
The gate $D(t)$ multiplies $\ket{x}$ by $\e^{2\pi\ii s_t(x)}$.  Writing
$P_{xy}:=\bra{x}P\ket{y}$, we therefore have
\begin{equation}
 \bigl(D(t)PD(t)^\dagger\bigr)_{xy}
 =\e^{2\pi\ii(s_t(x)-s_t(y))}P_{xy}.
 \label{eq:diagonal-conjugation-entry}
\end{equation}
Let
\[
 E_P:=\{(x,y)\in\{0,1\}^n\times\{0,1\}^n:P_{xy}\ne0\},
\]
choose any ordering of this finite set, and define
$A_P\in\Z^{|E_P|\times n}$ by the row indexed by $(x,y)$:
\[
 (A_P)_{(x,y),j}:=x_j-y_j.
\]
Equation~\eqref{eq:diagonal-conjugation-entry} gives the finite constraint
description
\begin{equation}
 D(t)PD(t)^\dagger=P
 \quad\Longleftrightarrow\quad
 A_Pt\in\Z^{|E_P|}.
 \label{eq:preservation-integer-matrix}
\end{equation}

We first identify the homogeneous real directions.  Let
$h\in\ker_{\R}A_P$ and set
\begin{equation}
 H_h:=\sum_{j=1}^n h_jN_j.
 \label{eq:diagonal-Hh}
\end{equation}
Since $H_h$ is diagonal with eigenvalue $h\cdot x$ on $\ket{x}$,
\[
 (H_hP-PH_h)_{xy}=((x-y)\cdot h)P_{xy}=0.
\]
Thus $H_hP=PH_h$.  Equation~\eqref{eq:Nj-detection} gives
\[
 PH_hP=\sum_{j=1}^n h_jPN_jP=(h\cdot\mu)P.
\]
Using $P^2=P$ and the commutation relation,
\[
 H_hP=H_hP^2=(H_hP)P=(PH_h)P=PH_hP=(h\cdot\mu)P.
\]
Taking powers and exponentiating yields
\begin{equation}
 D(h)P=\e^{2\pi\ii(h\cdot\mu)}P.
 \label{eq:kernel-direction-scalar}
\end{equation}
Hence every direction in $\ker_{\R}A_P$ acts on the code by a common scalar.

Apply Lemma~\ref{lem:uniform-rational-section} to $A_P$, and let $Q_0$ be
the resulting positive integer.  Take any $\theta\in\mathcal T_P$ and choose
a real lift $t\in\R^n$.  By Eq.~\eqref{eq:preservation-integer-matrix},
$A_Pt$ is integral, so there is a $c\in\Z^n$ such that
\[
 A_P\left(t-\frac{c}{Q_0}\right)=0.
\]
Set
\[
 r:=\frac{c}{Q_0},
 \qquad
 h:=t-r.
\]
Then $A_Ph=0$, while $A_Pr=A_Pt$ is integral, so $D(r)$ also preserves the
code.  By Eq.~\eqref{eq:kernel-direction-scalar},
\[
 D(t)P=D(r)D(h)P
 =\e^{2\pi\ii(h\cdot\mu)}D(r)P.
\]
Thus $D(t)$ and $D(r)$ induce the same projective logical operation.

The phase vector $r\bmod\Z^n$ depends only on $c\bmod Q_0$.  Therefore every
element of $\mathcal G_{\mathrm{diag}}(P)$ has a representative in the finite
set
\[
 \left(\frac1{Q_0}\Z/\Z\right)^n\cong(\Z_{Q_0})^n.
\]
It follows that $\mathcal G_{\mathrm{diag}}(P)$ is finite.  It is abelian
because all physical diagonal product gates commute.

Now let $\overline U$ be a nonidentity element of
$\mathcal G_{\mathrm{diag}}(P)$ with projective logical order $m$, and choose
one of the code-preserving representatives $r=c/Q_0$ constructed above.
Let $q$ be the smallest positive integer such that
\begin{equation}
 qr=0\qquad\text{in }\T^n.
 \label{eq:exact-phase-order}
\end{equation}
Such a $q$ exists because $Q_0r=0$.  Since $\overline U$ is nonidentity,
$r\ne0$, and hence $q\ge2$.

Since $qr_j=0\pmod{\Z}$ for every coordinate $j$, we may write
\[
 r_j=\frac{a_j}{q}\pmod\Z,
 \qquad 0\le a_j<q.
\]
Thus the corresponding physical gate is $D(a;q)$.  The phase parametrization
$D:\T^n\to U(2^n)$ is injective: applying $D(\theta)$ to the basis vector
with a single $1$ in position $j$ recovers the $j$th coordinate phase.
Consequently, Eq.~\eqref{eq:exact-phase-order} implies that $D(a;q)$ has
exact physical order $q$.

It remains to check that the representative is reduced.  If
$g:=\gcd(q,a_1,\ldots,a_n)>1$, then
\[
 \frac{q}{g}r
 =\left(\frac{a_1}{g},\ldots,\frac{a_n}{g}\right)
 =0\qquad\text{in }\T^n,
\]
contradicting the minimality of $q$.  Therefore
$\gcd(q,a_1,\ldots,a_n)=1$.  Finally, $D(a;q)^q=I$ implies
$\overline U^{\,q}=1$.  Since $\overline U$ has order $m$, we conclude that
$m\mid q$.
\end{proof}

\noindent
The theorem answers the existence question, but its mechanism is worth
making concrete.  The kernel step uses the same infinitesimal principle as
the Eastin--Knill argument reviewed in
Section~\ref{sec:transversal-no-go}: a code-preserving generator built from
detectable single-site terms acts as a scalar on the code.  Here we assume
only detection of $Z_j$, and the conclusion concerns gates diagonal in that
same fixed basis.  The following example serves two purposes.  It shows
explicitly how subtracting a homogeneous direction in $\ker_{\R}A_P$
replaces an irrational phase vector by a rational one, and it shows that the
minimum physical order $q_{\min}$ realizing a logical action can strictly
exceed its projective order $m$.

\begin{example}[Rational replacement and a physical-order gap]
\label{ex:seven-qubit-rationalization}
Consider the seven-qubit code spanned by
\begin{align}
 \ket{0_L}
 &=\frac{\ket{0000}+\ket{1111}}{\sqrt2}
   \otimes
   \frac{\ket{000}+\sqrt2\,\ket{111}}{\sqrt3},
 \label{eq:example-zero-L}\\
 \ket{1_L}
 &=\frac1{\sqrt6}\Bigl[
   (\ket{1100}+\ket{0011})\ket{110}
   +(\ket{1010}+\ket{1001})\ket{011}\notag\\
 &\hspace{28mm}
   +(\ket{0110}+\ket{0101})\ket{101}
   \Bigr].
 \label{eq:example-one-L}
\end{align}
Let \(T_0,T_1\) denote their computational-basis supports.  The squared
amplitudes define probability distributions on these disjoint supports
with the common first moment
\begin{equation}
 \mu=\left(\frac12,\frac12,\frac12,\frac12,
           \frac23,\frac23,\frac23\right).
 \label{eq:example-common-mean}
\end{equation}
Hence \(PZ_jP=(1-2\mu_j)P\) for every \(j\), so the hypotheses of
Theorem~\ref{thm:finite-projective-group} hold.  No two distinct strings in
\(T_0\cup T_1\) have Hamming distance one, so \(PX_jP=PY_jP=0\) as well.
On the other hand,
\[
 \bra{0_L}Z_1Z_2\ket{0_L}=1,
 \qquad
 \bra{1_L}Z_1Z_2\ket{1_L}=-\frac13.
\]
Thus the code has distance exactly two.

Because $T_0$ and $T_1$ are disjoint and every displayed coefficient is
nonzero, a diagonal gate preserves the code exactly when its phase is
constant on each support.  Comparing the strings within $T_1$ gives
\[
 \theta\equiv
 (\alpha+u,\alpha-u,\alpha,\alpha,
  \beta+u,\beta-u,\beta)\pmod1.
\]
Constancy on $T_0$, which contains the all-zero string, further requires
$4\alpha\equiv0$ and $3\beta\equiv0\pmod1$.  Set
\[
 v:=(1,-1,0,0,1,-1,0),
 \qquad
 \theta_{\alpha,\beta}
 :=(\alpha,\alpha,\alpha,\alpha,\beta,\beta,\beta).
\]
Thus
\begin{equation}
 \mathcal T_P
 =\left\{
 \theta_{\alpha,\beta}+uv\bmod\Z^7:
 \alpha\in\tfrac14\Z/\Z,\ 
 \beta\in\tfrac13\Z/\Z,\ 
 u\in\R
 \right\}.
 \label{eq:example-physical-group}
\end{equation}

Every string $x\in T_0\cup T_1$ satisfies
\[
 v\cdot x=x_1-x_2+x_5-x_6=0.
\]
Thus every row $x-y$ of the matrix $A_P$ in the proof of
Theorem~\ref{thm:finite-projective-group} annihilates $v$.  Conversely,
solving $A_Ph=0$ uses the same comparisons over $\R$, with the congruences
replaced by equalities.  These force $\alpha=\beta=0$ and leave only the
parameter $u$.  Therefore
\begin{equation}
 \ker_{\R}A_P=\R v.
 \label{eq:example-kernel}
\end{equation}
Moreover, $v\cdot\mu=0$.  Equation~\eqref{eq:kernel-direction-scalar}
therefore gives
\begin{equation}
 D(uv)P=P,
 \qquad u\in\R.
 \label{eq:example-kernel-identity}
\end{equation}
This is precisely the homogeneous direction removed in the rational-section
step of the theorem.  Subtracting $uv$ from any phase vector in
Eq.~\eqref{eq:example-physical-group} therefore gives the rational vector
$\theta_{\alpha,\beta}$ without changing its logical action.

The logical action of $\theta_{\alpha,\beta}$ is
\begin{equation}
 \diag\bigl(1,\e^{2\pi\ii(2\alpha+2\beta)}\bigr).
 \label{eq:example-logical-map}
\end{equation}
Writing $\alpha=a/4$ and $\beta=b/3$, its relative phase is encoded by the
homomorphism
\[
 \chi:\Z_4\times\Z_3\longrightarrow\Z_6,
 \qquad
 \chi(a,b)=3a+4b\pmod6.
\]
It is surjective because $\chi(1,1)=1$ in $\Z_6$, and solving
$\chi(a,b)=0$ gives the kernel $\{(0,0),(2,0)\}$.  Hence
\begin{equation}
 \mathcal G_{\mathrm{diag}}(P)
 \cong(\Z_4\times\Z_3)/\ker\chi
 \cong\Z_6.
 \label{eq:example-logical-group}
\end{equation}

Now define the rational phase vector
\begin{equation}
 \theta_0
 :=\frac1{12}(3,3,3,3,4,4,4),
 \label{eq:example-rational-theta-prime}
\end{equation}
and, for $t\in\R$, the family
\begin{equation}
 \theta(t):=\theta_0+tv
 =\left(
 \frac14+t,\frac14-t,\frac14,\frac14,
 \frac13+t,\frac13-t,\frac13
 \right).
 \label{eq:example-one-parameter-family}
\end{equation}
By Eq.~\eqref{eq:example-kernel-identity}, all these phase vectors induce
the same logical action,
\begin{equation}
 D(\theta(t))\big|_{\mathcal Q}
 =\diag(1,\e^{2\pi\ii/6}),
 \label{eq:example-logical-action}
\end{equation}
which has projective order $m=6$.  For irrational $t$, the displayed
physical representative has infinite order.  Removing the kernel direction
$tv$, exactly as in the proof of the theorem, gives $\theta_0$.  Since
$\gcd(12,3,3,3,3,4,4,4)=1$, $D(\theta_0)$ is in reduced form and has exact
physical order $q=12$.

It remains to show that twelve is minimal.  Let a finite-order phase vector
in Eq.~\eqref{eq:example-physical-group} induce the same projective logical
action.  Write $\alpha=a/4$ and $\beta=b/3$.  By
Eq.~\eqref{eq:example-logical-map}, the relative phase condition is
\begin{equation}
 3a+4b\equiv1\pmod6.
 \label{eq:example-target-congruence}
\end{equation}
Reducing this congruence modulo two shows that $a$ is odd, and reducing it
modulo three shows that $b\equiv1\pmod3$.  Hence $\alpha$ has exact order
four in $\T$, while $\beta$ has exact order three.  If the physical
representative has finite order $q$, then $q\theta=0$ in $\T^7$.  Its third
and seventh coordinates are $\alpha$ and $\beta$, respectively, so
$q\alpha=q\beta=0$.  Therefore $4\mid q$ and $3\mid q$, and hence
$12\mid q$.  Since $\theta_0$ attains this order,
\begin{equation}
 q_{\min}(\overline U)=12>6=m.
 \label{eq:example-minimum-order-gap}
\end{equation}
A search restricted to \(q=m\) would therefore miss this order-six
logical action in the chosen product basis.
\end{example}

We have now answered the first question posed at the beginning of the
section.  For each fixed code, there are finitely many projective diagonal
logical elements, and each has a finite-order physical representative.  One
problem remains: finiteness for each code does not yet give a uniform bound
on $q_{\min}(\overline U)$ in terms of $n$, $K$, or $d$.  We establish that
additional ingredient in Section~\ref{sec:q-bounds}.

An arbitrary displayed modulus is not an invariant: replacing $(q,a)$ by
$(rq,ra)$ leaves the local phases and subset-sum solution sets unchanged,
but introduces a common factor.  The later quantitative result bounds the
minimum physical order, not every presentation or every physical
representative of the same projective logical action.

\subsection{SSLP as a classification normal form}
\label{sec:sslp-normal-form}

Theorem~\ref{thm:finite-projective-group} gives us a finite-order physical
representative.  We next ask how much of the code that representative
determines.  The answer is exactly the SSLP support structure introduced in
Section~\ref{sec:sslp-preliminaries}: the logical eigenspaces must lie in
labeled subset-sum fibers.  We formulate this as a normal form for arbitrary
rank $K$.  For general $K$ the statement concerns a code together with a
chosen nontrivial projective diagonal logical element; the full finite
abelian group need not be cyclic.  The rank-two specialization is simpler
because one element can be chosen to generate the entire group.

\begin{theorem}[SSLP normal form for general code dimension]
\label{thm:general-sslp-normal-form}
Let $\mathcal Q$ be an exact $((n,K,d))$ code with $d\ge2$, and let
$\overline U$ be a nontrivial projective logical action induced by a product
gate diagonal in the chosen basis.  Let $m>1$ be the order of
$\overline U$.  Then, up to local $X$ conjugation, permutation of the
physical qubits, and a logical scalar, there are integers
\begin{equation}
 q\ge2,
 \qquad m\mid q,
 \qquad
 a\in\Z_q^n,
 \qquad
 0\le a_1\le\cdots\le a_n<q,
 \qquad
 \gcd(q,a_1,\ldots,a_n)=1,
 \label{eq:general-sslp-outer-data}
\end{equation}
distinct logical phase labels $\gamma_1,\ldots,\gamma_\ell\in\Z_m$, where
$\ell\ge2$, and positive multiplicities
$k_1,\ldots,k_\ell$ such that
\begin{equation}
 \gamma_1=0,
 \qquad
 \gcd(m,\gamma_2,\ldots,\gamma_\ell)=1,
 \qquad
 \sum_{\alpha=1}^{\ell}k_\alpha=K,
 \label{eq:general-sslp-colors}
\end{equation}
and the logical eigenspaces obey
\begin{equation}
 \dim E_\alpha=k_\alpha,
 \qquad
 E_\alpha\subseteq
 \Span\{\ket{x}:x\in S_{(q/m)\gamma_\alpha}(a;q)\}.
 \label{eq:general-colored-fibers}
\end{equation}

Conversely, suppose data satisfying
Eqs.~\eqref{eq:general-sslp-outer-data}--
\eqref{eq:general-sslp-colors} are given, and choose subspaces satisfying
Eq.~\eqref{eq:general-colored-fibers}.  If their direct sum satisfies the
Knill--Laflamme equations for all Pauli operators of weight less than a
target distance $d_0$, then their direct sum is a rank-$K$ code of distance
at least $d_0$, and $D(a;q)$ induces a projective logical action of exact
order $m$.
\end{theorem}

\begin{proof}
Since $d\ge2$, the code detects every single-qubit $Z_j$.
Theorem~\ref{thm:finite-projective-group} therefore gives a representative
$D(a;q)$ of $\overline U$ in reduced form, of exact physical order $q$,
with $m\mid q$.

\emph{Residue-sector decomposition.}
For $b\in\Z_q$, set
\[
 \mathcal H_b
 :=
 \Span\{\ket{x}:A_a(x)=b\},
 \qquad
 E_b:=\mathcal Q\cap\mathcal H_b.
\]
The spaces $\mathcal H_b$ are the eigenspaces of $D(a;q)$.  Since
$D(a;q)$ preserves $\mathcal Q$, its spectral projectors preserve
$\mathcal Q$, and hence
\[
 \mathcal Q=\bigoplus_{b\in\Z_q}E_b.
\]
Enumerate the residues for which $E_b\ne0$ as
$\widetilde b_1,\ldots,\widetilde b_\ell$, and put
$k_\alpha:=\dim E_{\widetilde b_\alpha}$.  Thus every $k_\alpha$ is
positive and
\[
 \sum_{\alpha=1}^{\ell}k_\alpha=K.
\]
The projective action is nontrivial, so $\ell\ge2$.

\emph{Canonical normalization.}
Choose a nonzero vector $\psi\in E_{\widetilde b_1}$, and choose
$x_0\in\{0,1\}^n$ with $\langle x_0\mid\psi\rangle\ne0$.  Then
$A_a(x_0)=\widetilde b_1$.  Write $t:=x_0$ and define
\[
 a^{(t)}_j:=(-1)^{t_j}a_j\pmod q.
\]
The one-qubit identity
\[
 X\diag(1,\omega_q^{a_j})X
 =
 \omega_q^{a_j}\diag(1,\omega_q^{-a_j})
\]
gives
\begin{equation}
 X^tD(a;q)X^t
 =
 \omega_q^{A_a(t)}D(a^{(t)};q),
 \qquad
 X^t:=\bigotimes_{j=1}^n X_j^{t_j}.
 \label{eq:general-local-X-gate}
\end{equation}
The corresponding subset sums satisfy
\begin{equation}
 A_{a^{(t)}}(x\mathbin{\oplus}t)
 =
 A_a(x)-A_a(t).
 \label{eq:general-local-X-residues}
\end{equation}
Consequently, after conjugating the code by $X^t$ and discarding the common
scalar in Eq.~\eqref{eq:general-local-X-gate}, the sector at
$\widetilde b_\alpha$ is carried to the sector at
\[
 b_\alpha:=\widetilde b_\alpha-\widetilde b_1,
\]
and $b_1=0$.  Sign changes of the coefficients preserve reducedness and
the exact physical order $q$; local $X$ conjugation preserves all Pauli
weights and hence the code distance.

With the common phase removed, the projective order is the least positive
integer $r$ such that
\[
 r b_\alpha=0\quad\text{in }\Z_q
 \qquad\text{for every }\alpha.
\]
Its exact value is $m$.  Since $m\mid q$, the elements of $\Z_q$ annihilated
by $m$ are precisely
\[
 \frac qm\gamma,
 \qquad \gamma\in\Z_m,
\]
with $\gamma$ unique.  Thus there is a unique
$\gamma_\alpha\in\Z_m$ such that
\[
 b_\alpha=\frac qm\gamma_\alpha.
\]
The $\gamma_\alpha$ are distinct because the $b_\alpha$ are distinct, and
$\gamma_1=0$.  The additive order of
$(\gamma_2,\ldots,\gamma_\ell)\in\Z_m^{\ell-1}$ is
\[
 \frac{m}{\gcd(m,\gamma_2,\ldots,\gamma_\ell)}.
\]
The embedding $\gamma\mapsto(q/m)\gamma$ preserves this order, so its value
is the projective logical order $m$.  This proves the gcd condition in
Eq.~\eqref{eq:general-sslp-colors}.

Finally, choose representatives of the entries of $a^{(t)}$ in
$\{0,\ldots,q-1\}$ and permute the physical qubits so that they are
nondecreasing.  Rename the resulting coefficient vector $a$ and the
transported sectors $E_\alpha$.  This permutation preserves reducedness,
the physical and projective orders, the dimensions $k_\alpha$, and the code
distance.  By construction,
\[
 E_\alpha
 \subseteq
 \Span\{\ket{x}:x\in
 S_{(q/m)\gamma_\alpha}(a;q)\},
\]
which proves all the forward assertions.  The inclusion need not be an
equality, so parametrizations over the full solution set must retain zero
coefficients.

\emph{Converse.}
Let $E_\alpha$ be subspaces satisfying
Eq.~\eqref{eq:general-colored-fibers}, and let $\Pi_\alpha$ be the orthogonal
projector onto $E_\alpha$.  The distinct labels $\gamma_\alpha$ give
distinct residues $(q/m)\gamma_\alpha$ in $\Z_q$.  Their ambient subset-sum
fibers are mutually orthogonal, so the $E_\alpha$ are mutually orthogonal.
Hence
\[
 P:=\sum_{\alpha=1}^{\ell}\Pi_\alpha
\]
is an orthogonal projector, with
\[
 \rank P
 =
 \sum_{\alpha=1}^{\ell}\rank\Pi_\alpha
 =
 \sum_{\alpha=1}^{\ell}k_\alpha
 =K.
\]
On $E_\alpha$, the gate $D(a;q)$ acts by the scalar
\[
 \omega_q^{(q/m)\gamma_\alpha}
 =
 \e^{2\pi\ii\gamma_\alpha/m}.
\]
It therefore preserves $\operatorname{im}P$.  Since $\gamma_1=0$, its
$r$th power acts as a common scalar on this direct sum exactly when
$r\gamma_\alpha=0$ in $\Z_m$ for every $\alpha$.  The least positive such
$r$ is
\[
 \frac{m}{\gcd(m,\gamma_2,\ldots,\gamma_\ell)}=m.
\]
Thus the induced projective logical action has exact order $m$.  Finally,
the assumed Knill--Laflamme equations for all Pauli operators of weight less
than $d_0$ say precisely that $P$ detects all such errors.  The resulting
rank-$K$ code therefore has distance at least $d_0$.
\end{proof}

Before specializing to rank two, we record why this normal form already
contains the LP layer of SSLP.  For a target distance $d_0$, every exact
completion produces the required moment data.  If
$\Pi_\alpha$ is the projector onto $E_\alpha$, then
\begin{equation}
 p_\alpha(x):=\frac{(\Pi_\alpha)_{xx}}{k_\alpha}
 \label{eq:general-sector-probabilities}
\end{equation}
is a probability distribution on the corresponding solution set, with zero
probabilities allowed.
For every $J\subseteq\{1,\ldots,n\}$ with $|J|<d_0$, the diagonal
Knill--Laflamme equations imply that
\begin{equation}
 \sum_x p_\alpha(x)\prod_{j\in J}x_j
 \label{eq:general-sector-moments}
\end{equation}
is independent of $\alpha$.  These linear moment equations provide a
preliminary test that cannot exclude a valid code.  They are therefore a
necessary projection of the full Knill--Laflamme system, not an independent
ansatz.  When an eigenspace has dimension greater than one, they do not
contain all within-sector matrix constraints, which remain part of the exact
KL completion.

For one logical qubit, the picture becomes especially transparent: the
general form has only two labeled solution sets, and one generator captures
the entire projective group.

\begin{corollary}[Rank-two SSLP normal form]
\label{cor:rank-two-sslp}
Let \(\mathcal Q\) be an exact \(((n,2,d))\) code with \(d\ge2\), and suppose
its projective diagonal logical group in the chosen product basis is
nontrivial.  Then that group is cyclic, say \(C_m\), and a generator has a
finite-order physical representative in reduced form \(D(a;q)\) for some
finite \(q\ge2\), with
\begin{equation}
 m\mid q.
 \label{eq:K2-bound}
\end{equation}
Up to local \(X\) conjugation, permutation of the physical qubits, a global
scalar, logical-state interchange, and replacement of the generator by its
\(u\)th power with
\(\gcd(u,q)=1\), there is an orthonormal logical eigenbasis of the form
\begin{align}
 \ket{0_L}&=\sum_{x\in S_0(a;q)}c_x\ket{x},
 \label{eq:0L-fiber}\\
 \ket{1_L}&=\sum_{y\in S_{q/m}(a;q)}d_y\ket{y},
 \label{eq:1L-fiber}
\end{align}
where zero coefficients are allowed and
\begin{equation}
 0\leq a_1\leq\cdots\leq a_n<q,
 \qquad
 \gcd(q,a_1,\ldots,a_n)=1.
 \label{eq:a-canonical}
\end{equation}
Conversely, any orthonormal pair of the form
\eqref{eq:0L-fiber}--\eqref{eq:1L-fiber} satisfying the full desired
Knill--Laflamme equations defines a code on which \(D(a;q)\) acts
projectively with order \(m\).
\end{corollary}

\begin{proof}
Because \(d\ge2\), the code detects every \(Z_j\).  By
Theorem~\ref{thm:finite-projective-group}, its projective diagonal logical
group \(G\) is finite abelian.

\emph{Cyclicity.}
Choose a nonidentity class $\overline V\in G$.  A unitary representative of
$\overline V$ has two distinct logical eigenlines $L_0,L_1$.  Physical
diagonal representatives commute, so every logical action in $G$ preserves
both lines.  If a
representative of $\overline W\in G$ acts on them by
$\lambda_0(W),\lambda_1(W)\in U(1)$, set
\[
 \chi(\overline W):=\frac{\lambda_1(W)}{\lambda_0(W)}.
\]
The ratio is unchanged when the representative is multiplied by a scalar,
so $\chi:G\to U(1)$ is a well-defined homomorphism.  Moreover,
$\chi(\overline W)=1$ exactly when $W$ acts as a common scalar on
$L_0\oplus L_1$, which is exactly the identity class in $G$.  Thus $\chi$
is injective.  Every finite subgroup of $U(1)$ is cyclic, so
$G\cong C_m$ for some $m\ge2$.

Choose a generator and apply
Theorem~\ref{thm:general-sslp-normal-form}.  A nontrivial logical unitary on
a two-dimensional space has two distinct eigenspaces, so \(\ell=2\) and
\(k_1=k_2=1\).  The two logical phase labels may be written as \(0\) and
\(\gamma\), where
\begin{equation}
 \gcd(\gamma,m)=1.
 \label{eq:rank-two-color}
\end{equation}

\emph{Unit lifting.}
We need a unit $u\in\Z_q^\times$ satisfying
$u\gamma\equiv1\pmod m$.  Write
\[
 q=\prod_p p^{e_p},
 \qquad
 m=\prod_p p^{f_p},
 \qquad 0\le f_p\le e_p.
\]
If $f_p>0$, choose a unit modulo $p^{e_p}$ that is congruent to
$\gamma^{-1}$ modulo $p^{f_p}$; if $f_p=0$, choose $1$ modulo $p^{e_p}$.
The Chinese remainder theorem combines these choices into a unit $u$ modulo
$q$, and the congruences for the primes dividing $m$ give
$u\gamma\equiv1\pmod m$.

Replace the physical generator by its \(u\)th power.  Because \(u\) is
relatively prime to $q$, its physical order remains $q$ and the
representative remains reduced.  Since $m\mid q$, it is also relatively
prime to $m$, so the logical action remains a generator.  Its coefficient
vector becomes $ua$ modulo $q$, while the second residue becomes
\[
 u\frac qm\gamma\equiv\frac qm\pmod q.
\]
Permute the qubits once more and choose coefficient representatives in
$\{0,\ldots,q-1\}$ to obtain Eq.~\eqref{eq:a-canonical}.  The two sectors
are one-dimensional and lie in distinct computational-basis fibers; choosing
a unit vector from each therefore gives the orthonormal basis in
Eqs.~\eqref{eq:0L-fiber}--\eqref{eq:1L-fiber}, with zero coefficients allowed.
Logical-state interchange is included among the harmless normalizations.

Conversely, the two displayed logical vectors are eigenvectors of \(D(a;q)\)
with eigenvalues
\[
 1,
 \qquad
 \omega_q^{q/m}=\e^{2\pi\ii/m}.
\]
Their eigenvalue ratio has exact order $m$, so the induced projective logical
action has exact order $m$.  If the two vectors are orthonormal and satisfy
the full Knill--Laflamme equations, their span is the required code.
\end{proof}

We can now see precisely what each SSLP layer contributes.  The SS layer is
forced by the code rather than imposed as an ansatz: it records the
finite-order physical representative, the subset-sum solution sets labeled
by logical eigenvalues, and the requirement that each logical eigenspace be
supported on its corresponding set.  The LP layer retains necessary linear
consequences of the diagonal Knill--Laflamme equations and provides a
preliminary test that cannot exclude a valid code.  Passing this test alone
does not guarantee that a code exists.  The final KL layer imposes the
complete matrix equations and is both necessary and sufficient for exact
realization.  Hence
\begin{equation}
 \boxed{
 \mathrm{SS}\;\bigl(\text{labeled subset-sum solution sets}\bigr)
 \ \longrightarrow\
 \mathrm{LP}\;\bigl(\text{diagonal moments}\bigr)
 \ \longrightarrow\
 \mathrm{KL}\;\bigl(\text{exact completion}\bigr)}
 \label{eq:three-layers}
\end{equation}
is an exact staged normal form for arbitrary \(K\), with
Corollary~\ref{cor:rank-two-sslp} giving the form with two solution sets used
below.

One obstacle remains before this normal form becomes a finite search: the
physical order \(q\) is known only to be finite and code dependent.
Section~\ref{sec:q-bounds} supplies a uniform cutoff at fixed parameters.
Searching only \(q=m\) is not exhaustive: the structural result guarantees
\(m\mid q\), and the minimum physical order can be strictly larger than the
projective logical order.

\subsection{Bounds on the physical order}
\label{sec:q-bounds}

The only remaining obstruction to a finite outer search is the absence of a
code-independent cutoff for $q$.  We now derive one from the geometry of the
logical phase sectors.  Their common first moments let us replace the
physical phases by a rational solution of an integral linear system, and
binary denominator bounds then control its physical order.  For \(K=2\),
the second moments available at distance at least three sharpen the bound.

We begin with the first-moment input.  Single-site detection forces every
logical phase sector to have the same barycenter, even when the sector has
dimension greater than one.

Suppose $D(\theta)$ preserves the code, with distinct logical eigenvalues
$\e^{2\pi\ii\phi_1},\ldots,\e^{2\pi\ii\phi_\ell}$.  Let $E_\alpha$ be the
corresponding logical eigenspaces, $\Pi_\alpha$ their orthogonal projectors,
and $k_\alpha=\rank\Pi_\alpha$.  Define their computational-basis supports by
\begin{equation}
 T_\alpha:=\{x\in\{0,1\}^n:(\Pi_\alpha)_{xx}>0\}.
 \label{eq:Talpha}
\end{equation}
Each support is nonempty because $\tr\Pi_\alpha=k_\alpha>0$.
As in the proof of Theorem~\ref{thm:general-sslp-normal-form}, the diagonal
entries of $D(\theta)\Pi_\alpha=\e^{2\pi\ii\phi_\alpha}\Pi_\alpha$ give
\begin{equation}
 \theta\cdot x\equiv\phi_\alpha\pmod1,
 \qquad x\in T_\alpha.
 \label{eq:phase-sector-congruence}
\end{equation}
The supports are therefore pairwise disjoint.  Every vector in $E_\alpha$
is supported on $T_\alpha$, since $(\Pi_\alpha)_{xx}$ is the sum of squared
$x$-amplitudes of an orthonormal basis of $E_\alpha$.

\begin{lemma}[Common first moments of logical phase sectors]
\label{lem:phase-sectors}
Suppose in addition that $PZ_jP=\lambda_jP$ for every $j$, and set
$\mu_j=(1-\lambda_j)/2$.  For $x\in T_\alpha$, define
\begin{equation}
 p_\alpha(x):=\frac{(\Pi_\alpha)_{xx}}{k_\alpha}.
 \label{eq:palpha}
\end{equation}
Then $p_\alpha$ is a positive probability distribution on $T_\alpha$, and
all these distributions have the same first moment:
\begin{equation}
 \sum_{x\in T_\alpha}p_\alpha(x)x=\mu,
 \qquad \mu=(\mu_1,\ldots,\mu_n).
 \label{eq:common-mean}
\end{equation}
\end{lemma}

\begin{proof}
Positivity follows from the definition of $T_\alpha$, and
\[
 \sum_{x\in T_\alpha}p_\alpha(x)
 =\frac{\tr\Pi_\alpha}{k_\alpha}=1.
\]
The detection assumption is equivalent to $PN_jP=\mu_jP$.  Since
$P\Pi_\alpha=\Pi_\alpha P=\Pi_\alpha$, it follows that
$\Pi_\alpha N_j\Pi_\alpha=\mu_j\Pi_\alpha$.  Taking traces gives
\[
 \sum_{x\in T_\alpha}x_j(\Pi_\alpha)_{xx}=\mu_jk_\alpha.
\]
Dividing by $k_\alpha$ proves Eq.~\eqref{eq:common-mean}.
\end{proof}

At first sight, irrational sector probabilities might seem to obstruct a
rational phase replacement.  They do not: the common barycenter is the only
consequence of single-site detection that the argument needs.

\begin{lemma}[Common-mean rational replacement]
\label{lem:common-mean-rational-replacement}
Let $T_1,\ldots,T_\ell\subseteq\{0,1\}^n$ carry probability
distributions $p_\alpha$ with a common first moment
\[
 \sum_{x\in T_\alpha}p_\alpha(x)x=\mu.
\]
Suppose every vector in $E_\alpha$ is supported on $T_\alpha$, and let
$\theta,\theta'\in\R^n$.  Set $\Delta=\theta'-\theta$.  If
\[
 \Delta\cdot(x-y)=0
 \qquad
 \text{for all $x,y\in T_\alpha$ and every $\alpha$,}
\]
then
\[
 \Delta\cdot x=\Delta\cdot\mu
 \qquad
 \text{for every $x\in\bigcup_{\alpha=1}^{\ell}T_\alpha$.}
\]
Consequently,
\[
 \left.D(\theta')\right|_{\bigoplus_\alpha E_\alpha}
 =
 \e^{2\pi\ii\,\Delta\cdot\mu}
 \left.D(\theta)\right|_{\bigoplus_\alpha E_\alpha}.
\]
In particular, if $D(\theta)$ preserves
$\bigoplus_\alpha E_\alpha$, then so does $D(\theta')$, and the two
gates induce the same projective action there.
\end{lemma}

\begin{proof}
For each $\alpha$, the hypothesis says that $\Delta\cdot x$ has a
constant value $c_\alpha$ on $T_\alpha$.  Averaging that constant with
respect to $p_\alpha$ gives
\[
 c_\alpha
 =\sum_{x\in T_\alpha}p_\alpha(x)\,\Delta\cdot x
 =\Delta\cdot
   \left(\sum_{x\in T_\alpha}p_\alpha(x)x\right)
 =\Delta\cdot\mu.
\]
Thus the constants are independent of $\alpha$.  For every
$v\in E_\alpha$, the diagonal ratio of the two gates is constant on the
support of $v$, and hence
\[
 D(\theta')v=
 \e^{2\pi\ii\,\Delta\cdot\mu}D(\theta)v.
\]
The claimed equality follows by taking the direct sum over the logical
eigenspaces.
\end{proof}

This is the key economy of the argument: no rationality of the probabilities
is required, and no weight-two Knill--Laflamme equation is used.  Second
moments enter only in the distance-three refinement below.

Having reduced rational replacement to a linear problem, we now need a
uniform denominator bound.  We obtain it from a binary formulation of the
phase equations.  After normalizing the first logical phase to zero, the
system has $n$ physical coordinates and $\ell-1$ auxiliary phase coordinates.
Selecting independent equations and variables gives an invertible binary
matrix of size at most $n+\ell-1$, whose Smith normal form controls the
denominators of the solution.

Here a binary matrix has entries zero or one, but all ranks and inverses
are computed over $\mathbb R$, equivalently $\mathbb Q$, rather than
modulo two.  For an invertible binary matrix $C\in\{0,1\}^{s\times s}$,
its \emph{Smith normal form} is
\begin{equation}
 UCV=\diag(d_1(C),\ldots,d_s(C)),
 \qquad
 1\leq d_1(C)\mid\cdots\mid d_s(C),
 \label{eq:snf}
\end{equation}
where $U,V$ are integer matrices with determinant $+1$ or $-1$.  The numbers
$d_i(C)$ are its Smith invariant factors.  The final factor $d_s(C)$
provides a common multiplier that clears the denominators of solutions.
Equivalently, it is the exponent of the finite cokernel
$\Z^s/C\Z^s$.  The following lemma gives the form used below.

\begin{lemma}[Elementary denominator clearing]
\label{lem:denominator-clearing}
If $Cz=h$ with $h\in\Z^s$, then
\begin{equation}
 d_s(C)z\in\Z^s.
 \label{eq:clear-denom}
\end{equation}
Consequently, the least common denominator of the coordinates of $z$ divides
$d_s(C)$.
\end{lemma}

\begin{proof}
The Smith decomposition gives
\[
 C^{-1}=V\diag(d_1(C)^{-1},\ldots,d_s(C)^{-1})U.
\]
Since every $d_i(C)$ divides $d_s(C)$, the matrix $d_s(C)C^{-1}$ is integral.
Applying it to $h$ proves the claim.
\end{proof}

For a positive integer $\nu$, define
\begin{equation}
 \beta_\nu:=
 \max_{1\leq s\leq\nu}
 \max_{\substack{C\in\{0,1\}^{s\times s}\\ \det C\neq0}}
 d_s(C).
 \label{eq:beta}
\end{equation}
Thus $\beta_\nu$ bounds the multiplier needed to clear solution
denominators for every invertible binary system with at most $\nu$
variables and integer right-hand side.

We are now ready to state the general cutoff.

\begin{theorem}[Physical-order bound for general code dimension]
\label{thm:colored-finite-lift}
Let $P$ be a rank-$K$ code projector satisfying
\[
 PZ_jP=\lambda_jP,
 \qquad j=1,\ldots,n.
\]
Suppose $D(\theta)$ preserves the code and induces a nontrivial projective
logical action $\overline U$, with order $m$ and $\ell\ge2$ distinct logical
eigenvalues.  Then the same projective action has a physical representative in
reduced form $D(a;q)$ satisfying
\begin{equation}
 m\mid q,
 \qquad
 q\leq\beta_{n+\ell-1}\leq\beta_{n+K-1}.
 \label{eq:general-q-bound}
\end{equation}
More precisely, the representative can be obtained from an invertible
binary matrix $C\in\{0,1\}^{s\times s}$ such that
\begin{equation}
 q\mid d_s(C),
 \qquad
 1\le s\le n+\ell-1.
 \label{eq:q-divides-smith}
\end{equation}
In particular,
\begin{equation}
 q_{\min}(\overline U)\leq\beta_{n+\ell-1}.
 \label{eq:qmin-general}
\end{equation}
\end{theorem}

\begin{proof}
Let $T_1,\ldots,T_\ell$ be the logical phase supports in
Eq.~\eqref{eq:Talpha}.  By Lemma~\ref{lem:phase-sectors}, their probability
distributions have a common first moment $\mu$.

We first normalize the first sector so that its phase requires no auxiliary
variable.  Choose $x_0\in T_1$ and relabel the computational basis by
\[
 x\longmapsto x\oplus x_0.
\]
This is implemented by conjugating both the code and the gate by $X^{x_0}$.
Apart from a common scalar, the physical phases change by
\begin{equation}
 \theta_j\longmapsto(-1)^{(x_0)_j}\theta_j.
 \label{eq:X-conjugation}
\end{equation}
On the $j$th coordinate this relabeling is the affine map
$x_j\mapsto x_j$ or $x_j\mapsto1-x_j$.  It therefore sends the common
first moment of all sector distributions to another common first moment.
Single-site detection is preserved as well.  Discard the common scalar and
rename the transformed data.  Since the first support now contains $0^n$,
and $D(\theta)\ket{0^n}=\ket{0^n}$, we have
\begin{equation}
 0^n\in T_1,
 \qquad
 \phi_1=0.
 \label{eq:normalized-first-sector}
\end{equation}

Choose real lifts of the phase congruences.  They take the form
\begin{align}
 \theta\cdot x&=h_x,
 &&x\in T_1,
 \label{eq:first-sector-eq}\\
 \theta\cdot x-\phi_\alpha&=h_x,
 &&x\in T_\alpha,\quad \alpha=2,\ldots,\ell,
 \label{eq:other-sector-eq}
\end{align}
where $h_x\in\Z$.  Introduce the auxiliary phase coordinates and write
\[
 u=(\theta_1,\ldots,\theta_n,-\phi_2,\ldots,-\phi_\ell)^T.
\]
Equations~\eqref{eq:first-sector-eq}--\eqref{eq:other-sector-eq} then form
the consistent integral system
\begin{equation}
 Ru=h,
 \label{eq:Ru=h}
\end{equation}
where $R$ is binary and has $n+\ell-1$ columns.

Let $s=\rank_{\Q}R$, equivalently $\rank_{\R}R$.  Here $s\ge1$:
since $\ell\ge2$, the nonempty support $T_2$ supplies a row of $R$, and
that row has entry $1$ in the auxiliary coordinate corresponding to
$-\phi_2$.  Thus
\[
 1\le s\le n+\ell-1.
\]

Select $s$ linearly independent rows of $R$, indexed by a set $I$, and
then select $s$ columns, indexed by a set $J$, such that the square
submatrix
\[
 C=R_{I,J}\in\{0,1\}^{s\times s}
\]
is invertible.  Let $h_I$ be the corresponding part of the right-hand
side.  Define $u'$ by setting its coordinates outside $J$ equal to zero
and solving
\begin{equation}
 Cu'_J=h_I.
 \label{eq:pivot-solve}
\end{equation}
Since $C$ and $h_I$ are integral, $u'$ is rational.

We verify explicitly that $u'$ solves the full system.  Each omitted row
$r$ of $R$ is a linear combination
\[
 r=\sum_{i\in I}c_iR_i
\]
of the selected rows.  Because the original system is consistent, its
solution $u$ gives
\[
 h_r=ru=\sum_{i\in I}c_iR_i u
     =\sum_{i\in I}c_i h_i.
\]
The selected equations hold for $u'$ by construction, so the same
calculation gives
\[
 ru'=\sum_{i\in I}c_iR_i u'
    =\sum_{i\in I}c_i h_i
    =h_r.
\]
Hence $Ru'=h$.

Lemma~\ref{lem:denominator-clearing} applied to
Eq.~\eqref{eq:pivot-solve} gives
\[
 d_s(C)u'_J\in\Z^s.
\]
The remaining coordinates of $u'$ vanish, and therefore
\[
 d_s(C)u'\in\Z^{n+\ell-1}.
\]

Let $\theta'$ be the physical part of $u'$ and set
$\Delta=\theta'-\theta$.  For $\alpha\ge2$, write $-\phi'_\alpha$ for
the corresponding auxiliary coordinate of $u'$, and set
$\gamma_1=0$ and $\gamma_\alpha=\phi'_\alpha-\phi_\alpha$.  Because
$u$ and $u'$ solve the same lifted system, subtracting their equations
at $x\in T_\alpha$ first gives
\[
 \Delta\cdot x=\gamma_\alpha.
\]
Subtracting this equality at two points of the same sector then gives
\[
 \Delta\cdot(x-y)=0,
 \qquad x,y\in T_\alpha.
\]
Lemma~\ref{lem:common-mean-rational-replacement} now gives
\begin{equation}
 \Delta\cdot x=\Delta\cdot\mu
 \qquad\text{for every }x\in\bigcup_\alpha T_\alpha,
 \label{eq:common-phase-shift}
\end{equation}
and hence
\begin{equation}
 D(\theta')\big|_{\mathcal Q}
 =\e^{2\pi\ii\,\Delta\cdot\mu}D(\theta)\big|_{\mathcal Q}.
 \label{eq:rational-replacement}
\end{equation}
In the present normalization this scalar is $1$: the left side of
Eq.~\eqref{eq:common-phase-shift} vanishes at $0^n\in T_1$, so
$\Delta\cdot\mu=0$.  Thus $D(\theta')$ preserves the normalized code and
induces exactly the same logical action there.

Let $q$ be the least positive integer for which $q\theta'\in\Z^n$.  The
integrality of $d_s(C)u'$ and the elementary least-common-denominator
property imply
\[
 q\mid d_s(C).
\]
Put $\widetilde a_j=q\theta'_j\in\Z$, and let $a_j$ be its residue modulo
$q$.  Then
\[
 D(\theta')=D(a;q).
\]
The minimality of $q$ shows that this gate has exact physical order $q$.
It also gives the gcd-one condition: if an integer $g>1$ divided $q$ and
every canonical representative of $a_j$, then it would divide every
$\widetilde a_j$, because $\widetilde a_j\equiv a_j\pmod q$.  Hence
\[
 \frac qg\,\theta'_j=\frac{\widetilde a_j}{g}\in\Z
 \qquad\text{for every }j,
\]
contrary to the minimality of $q$.

The induced projective logical action has order $m$.  Since the $q$th
power of the physical gate is the identity, its projective restriction
also has $q$th power equal to the identity.  Therefore $m\mid q$.
Nontriviality gives $2\le m\le q$, so $q\ge2$; together with the gcd-one
condition, this proves that $D(a;q)$ is in reduced form.

Undoing the basis relabeling changes the physical phases only by the sign
changes in Eq.~\eqref{eq:X-conjugation}, together with a common scalar.
Sign changes preserve the least common denominator and the reducedness
condition, while the scalar does not affect the projective logical action.
Thus the resulting representative of the original action still has exact
physical order $q$.

Finally,
\[
 q\le d_s(C)
 \le\beta_{n+\ell-1}
 \le\beta_{n+K-1},
\]
where the first inequality follows from $q\mid d_s(C)$, the second uses
$1\le s\le n+\ell-1$, and the last uses $\ell\le K$.
Equation~\eqref{eq:qmin-general} follows because
$q_{\min}(\overline U)$ is the minimum over all physical representatives
of the same projective action.
\end{proof}

This closes the general finite-search argument.  Every exact code of distance
at least two satisfies the detection hypothesis, so the theorem, together
with the SSLP normal form of Section~\ref{sec:sslp-normal-form}, supplies a
finite range of physical representative orders at fixed $n$ and $K$.

Taking the largest final Smith invariant factor among the nonsingular
binary matrices listed in Tables~A.3 and~A.5 of
Ref.~\cite{Zivkovic} gives
\begin{equation}
 \beta_5=5,
 \qquad
 \beta_6=9,
 \qquad
 \beta_7=18,
 \qquad
 \beta_8=42,
 \qquad
 \beta_9=120.
 \label{eq:small-beta}
\end{equation}
For \(K=2\), every nontrivial projective diagonal action has \(\ell=2\), so
Theorem~\ref{thm:colored-finite-lift} gives a reduced physical representative
of the same projective action satisfying
\begin{equation}
 m\mid q,
 \qquad
 2\le q\le\beta_{n+1}.
 \label{eq:K2-q-bound-summary}
\end{equation}
For six qubits this gives \(q\le18\).  Distance-three geometry improves the
bound by one matrix dimension, lowering the cutoff to \(q\le9\), as we now
show.  The improvement comes from a consequence of the second moments: the
two intrinsic supports must have the same affine span.

For a nonempty set $S\subseteq\R^n$, write $\aff(S)$ for its affine span.
Equivalently, for any $x_0\in S$,
\[
 \aff(S)=x_0+\spanR\{x-x_0:x\in S\},
\]
which is independent of the choice of $x_0$.

\begin{lemma}[Positive affine moments determine the affine span]
\label{lem:positive-affine-moments}
Let $T_0,T_1\subset\R^n$ be finite nonempty sets, and let $p_k$ be a
strictly positive probability distribution on $T_k$.  Suppose that, for
every affine function $f(x)=c+w\cdot x$,
\begin{equation}
 \sum_{x\in T_0}p_0(x)(c+w\cdot x)^2
 =\sum_{y\in T_1}p_1(y)(c+w\cdot y)^2.
 \label{eq:affine-square-moments}
\end{equation}
Then
\begin{equation}
 \aff(T_0)=\aff(T_1).
 \label{eq:same-affine-span}
\end{equation}
\end{lemma}

\begin{proof}
Because every weight is strictly positive, the left-hand side of
Eq.~\eqref{eq:affine-square-moments} vanishes if and only if $f$ vanishes
on all of $T_0$.  The same statement holds on the right for $T_1$.
Thus an affine function vanishes on $T_0$ if and only if it vanishes on
$T_1$.  The affine span of a set is the common zero set of all affine
functions that vanish on that set, so the two affine spans are equal.
\end{proof}

Applying this affine-span identity to the phase equations yields the sharper
rank-two bound.

\begin{corollary}[Distance-three physical-order bound for rank-two codes]
\label{cor:d3-physical-order-bound}
Let $\mathcal Q$ be an exact $((n,2,d))$ code with $d\ge3$, and let
$D(\theta)$ be a product gate diagonal in the chosen basis that preserves
$\mathcal Q$ and induces a nontrivial projective logical action
$\overline U$ of order $m$.  The restriction of $D(\theta)$ to
$\mathcal Q$ has two distinct one-dimensional eigenspaces $E_0,E_1$;
let $\Pi_0,\Pi_1$ be their orthogonal projectors and define the intrinsic
supports
\[
 T_k:=\{x\in\{0,1\}^n:(\Pi_k)_{xx}>0\},
 \qquad k=0,1.
\]
If $\ket{\psi_k}$ is any normalized vector spanning $E_k$, then
$(\Pi_k)_{xx}=|\langle x\mid\psi_k\rangle|^2$.  Thus $T_k$ is also the
ordinary nonzero computational-basis support of a logical eigenvector.
Then
\[
 \aff(T_0)=\aff(T_1).
\]
If $r$ is their common affine dimension, then $1\le r\le n$, and the
same projective action has a physical representative in reduced form
$D(a;q)$ of exact order $q$.  Moreover, after a local $X$ normalization,
there is an invertible binary matrix $C\in\{0,1\}^{r\times r}$, obtained
as a coordinate minor of a matrix whose rows are a basis selected from
one of the two supports, such that
\begin{equation}
 m\mid q,
 \qquad
 q\mid d_r(C),
 \qquad
 2\le q\le d_r(C)\le\beta_r\le\beta_n.
 \label{eq:d3-full-bound}
\end{equation}
Here invertibility and rank are over $\Q$, equivalently over $\R$; no
rank is taken over $\mathbb F_2$.  In particular,
\begin{equation}
 q_{\min}(\overline U)\le\beta_r\le\beta_n.
 \label{eq:d3-qmin}
\end{equation}
\end{corollary}

\begin{proof}
For $k=0,1$, set
\[
 p_k(x):=(\Pi_k)_{xx},\qquad x\in T_k.
\]
Because $\Pi_k$ has rank one, $p_k$ is a strictly positive probability
distribution on $T_k$.  The diagonal Knill--Laflamme equations for errors
of weight at most two give equality of the first and second moments of
$p_0$ and $p_1$.  Expanding $(c+w\cdot x)^2$, and using $x_j^2=x_j$ for
$x\in\{0,1\}^n$, therefore gives Eq.~\eqref{eq:affine-square-moments} for
every $c\in\R$ and $w\in\R^n$.  Lemma~\ref{lem:positive-affine-moments}
now gives Eq.~\eqref{eq:same-affine-span}.

Choose $x_0\in T_0$.  The local bit flip $X^{x_0}$ acts on bit strings as
the restriction of the affine involution $F_{x_0}:\R^n\to\R^n$ defined by
\[
 (F_{x_0}(x))_j=
 \begin{cases}
  x_j,&(x_0)_j=0,\\
  1-x_j,&(x_0)_j=1.
 \end{cases}
\]
Thus $F_{x_0}(x)=x\oplus x_0$ on $\{0,1\}^n$.  Its linear part is the
invertible diagonal map with entries $(-1)^{(x_0)_j}$, so $F_{x_0}$ sends
affine spans to affine spans and preserves their dimensions.  Conjugate
the code and gate by $X^{x_0}$, discard the resulting common scalar, and
rename the transformed data.  Then $0^n\in T_0$, the eigenvalue on $E_0$
is $1$, and the common affine span is a linear subspace $V$ of dimension
$r$.

The supports $T_0$ and $T_1$ are nonempty and disjoint: their phase-sector
congruences correspond to the two distinct logical eigenvalues.  If $r=0$,
then their common affine span would be a single point, forcing both nonempty
supports to equal that point.  This contradicts disjointness.  Hence
$1\le r\le n$.

Since $0^n\in T_0$ and $\aff(T_0)=V$, the set $T_0$ spans $V$ linearly.
Choose $v_1,\ldots,v_r\in T_0$ forming a basis of $V$, and let
$B\in\{0,1\}^{r\times n}$ be the matrix with these vectors as rows.  Its
row rank is $r$ over $\R$, and hence also over $\Q$ because $B$ has rational
entries.  Select coordinates $I\subseteq\{1,\ldots,n\}$, $|I|=r$, such
that the coordinate minor
\[
 C:=B_{[:,I]}\in\{0,1\}^{r\times r}
\]
has nonzero determinant.  This is nonsingularity over $\Q$ and $\R$; the
matrix need not be nonsingular modulo two.

For each $v_i\in T_0$, the normalized gate satisfies
$D(\theta)\ket{v_i}=\ket{v_i}$, so
$h:=B\theta\in\Z^r$.  Define $\theta'\in\Q^n$ by
\begin{equation}
 \theta'_I=C^{-1}h,
 \qquad
 \theta'_j=0\quad(j\notin I).
 \label{eq:C-theta-prime}
\end{equation}
Then $B\theta'=h=B\theta$.  Therefore $\theta'-\theta$ annihilates every
row of $B$, hence all of $V$.  Since both $T_0$ and $T_1$ lie in $V$,
\[
 (\theta'-\theta)\cdot x=0
 \qquad (x\in T_0\cup T_1).
\]
Thus $D(\theta')$ and $D(\theta)$ agree pointwise on both intrinsic
supports.  Every vector in $E_k$ is supported on $T_k$, so the two gates
have the same restriction on $E_0\oplus E_1=\mathcal Q$.  In particular,
$D(\theta')$ preserves the code and induces exactly the same projective
logical action.

By Lemma~\ref{lem:denominator-clearing},
$d_r(C)\theta'_I\in\Z^r$; the coordinates outside $I$ vanish, so
$d_r(C)\theta'\in\Z^n$.  Let $q$ be the least positive integer such that
$q\theta'\in\Z^n$.  Then $q\mid d_r(C)$, and the minimality of $q$ shows
that reducing $q\theta'$ modulo $q$ gives a phase vector of exact physical
order $q$ with the gcd-one property.  The order $m$ of its projective
restriction divides $q$.  Since this restriction is nontrivial, $m>1$ and
hence $2\le m\le q$; in particular $q\ge2$, so this is a reduced
representative $D(a;q)$.  Consequently,
\[
 m\mid q,
 \qquad
 q\mid d_r(C),
 \qquad
 2\le q\le d_r(C)\le\beta_r\le\beta_n,
\]
which proves Eq.~\eqref{eq:d3-full-bound}.

Undoing the local $X$ normalization applies the inverse affine involution
and changes the physical phase coordinates only by signs, up to a common
scalar.  It therefore preserves the affine dimension, the exact physical
order $q$, and the desired projective logical action.  Finally, the
constructed representative is included in the definition of
$q_{\min}(\overline U)$, so Eq.~\eqref{eq:d3-qmin} follows.
\end{proof}

We can now state the promised finite-enumeration consequence.  Combining
Corollaries~\ref{cor:rank-two-sslp}
and~\ref{cor:d3-physical-order-bound} makes the rank-two SSLP enumeration
finite.  For a fixed length \(n\) and target distance \(d_0\ge2\), it
suffices to consider the reduced data
\begin{equation}
 m\mid q,\qquad m>1,\qquad
 2\le q\le
 \begin{cases}
  \beta_{n+1},&d_0=2,\\
  \beta_n,&d_0\ge3,
 \end{cases}
 \label{eq:K2-finite-cutoff}
\end{equation}
with \(a\) normalized as in Eq.~\eqref{eq:a-canonical}.  The local \(X\)
normalization, coordinate permutations, and unit-power generator
normalization in that corollary preserve the reduced physical order.
For each choice, use the two logical vectors in
Eqs.~\eqref{eq:0L-fiber}--\eqref{eq:1L-fiber}, allowing zero coefficients,
and impose the complete Knill--Laflamme equations for every Pauli operator
of weight less than \(d_0\).  Every code in the stated sector occurs in
this enumeration, and every orthonormal solution defines such a code with
a projective logical element of order \(m\).  At distance at least three,
Corollary~\ref{cor:d3-physical-order-bound}, specifically
Eq.~\eqref{eq:d3-full-bound}, also supplies the sharper restriction
\(q\le\beta_r\) based on the dimension of the actual supports.
There are finitely many data \((q,m,a)\), and the completion conditions
are polynomial equations in the real and imaginary parts of the amplitudes.

The bounds apply to every nontrivial projective diagonal logical element.
For a fixed rank-two code, we can therefore recover the entire projective
diagonal logical group: test all reduced representatives within the
corresponding cutoff in the fixed product basis, retain those preserving the
code, identify equal projective restrictions, and adjoin the identity.

\section{Classification of six-qubit transversal diagonal symmetries}
\label{six:sec:scope}

Section~\ref{sec:sslp-classification} reduces the six-qubit problem to a
finite search over diagonal symmetry data.  For exact \(((6,2,3))\) codes,
every nontrivial projective diagonal logical action admits a physical
representative of order \(2\le q\le9\).  The SSLP search yields \(66\)
candidate classes of symmetry data, of which exactly \(15\) admit completions
satisfying the full Knill--Laflamme equations.  This section summarizes the
resulting classification, weight-enumerator results, and uniqueness of the
order-five code up to local unitaries and qubit permutations.  Throughout,
diagonal actions refer to the fixed product basis.  The companion GitHub
repository~\cite{SSLPClassificationSupplement} contains the Supplementary
Materials with detailed proofs, together with the accompanying code, data,
exact certificates, and Lean proof sources.

\subsection{Classification task and equivalences}

We first specialize the rank-two SSLP normal form to six qubits.
Corollary~\ref{cor:d3-physical-order-bound} supplies a physical
representative of every nontrivial diagonal logical action within the
cutoff
\begin{equation}
  m\mid q,
  \qquad
  2\le q\le\beta_6=9.
  \label{six:eq:qcutoff}
\end{equation}
where \(\beta_6=9\) is given in Eq.~\eqref{eq:small-beta}.
Using the notation of Section~\ref{sec:sslp-preliminaries}, write
\[
  \omega_q=e^{2\pi i/q},
\]
and
\begin{equation}
  D(a;q)
  :=
  \bigotimes_{j=1}^6\diag(1,\omega_q^{a_j}).
  \label{six:eq:Daq}
\end{equation}
We retain reduced representatives, whose exact physical order is \(q\):
\begin{equation}
  \gcd(q,a_1,\ldots,a_6)=1.
  \label{six:eq:primitive}
\end{equation}
For \(x\in\{0,1\}^6\), their subset-sum fibers are
\begin{equation}
  A_a(x):=\sum_{j=1}^6a_jx_j,
  \qquad
  S_b(a;q)
  :=
  \{x\in\{0,1\}^6:A_a(x)\equiv b\pmod q\}.
  \label{six:eq:fibers}
\end{equation}
Using the normalizations of Corollary~\ref{cor:rank-two-sslp}, we write
the logical eigenvectors in the two residue fibers \(S_0(a;q)\) and
\(S_{q/m}(a;q)\):
\begin{align}
  \ket{0_L}
  &=
  \sum_{x\in S_0(a;q)}c_x\ket{x},
  \label{six:eq:sslp0}\\
  \ket{1_L}
  &=
  \sum_{y\in S_{\delta}(a;q)}d_y\ket{y},
  \qquad
  \delta=\frac qm,
  \label{six:eq:sslp1}
\end{align}
with
\begin{equation}
  0\le a_1\le\cdots\le a_6<q.
  \label{six:eq:sorted}
\end{equation}
Zero coefficients are allowed, so the actual supports can be proper subsets
of the full subset-sum solution sets.

\begin{definition}[SSLP candidate classes and realizable classes]
\label{def:six-sslp-orbits}
Fix a physical representative in reduced form, specified by \((q,m,a)\)
and satisfying Eq.~\eqref{six:eq:qcutoff}.  Consider labeled data
\((q,m;S_0,S_{q/m})\) admitting probability distributions with matching
first and second moments.  An \emph{SSLP candidate class} consists of all
such data related by bit flips, coordinate permutations, and interchange
of the two logical phase labels.  Both solution sets are taken in full,
and the labels \((q,m)\) are retained.  A class \emph{satisfies the
support-dimension bound} if these distributions can be chosen so that
their supports have a common affine span of dimension
\(\rho\), with \(q\le\beta_\rho\).  The exact test for this condition is
given under ``Active-rank reduction'' in the Supplementary
Materials~\cite{SSLPClassificationSupplement}.

For a candidate class \(r\), let \(\mathcal K_r\) be the set of pairs of
codewords of the form
Eqs.~\eqref{six:eq:sslp0}--\eqref{six:eq:sslp1} that satisfy the full
weight-at-most-two Knill--Laflamme equations, including all non-diagonal
amplitude and phase constraints.  The class is \emph{realizable} when
\begin{equation}
  \mathcal K_r\ne\varnothing.
  \label{eq:Kr-definition}
\end{equation}
Thus the subset-sum and moment stages produce candidates, and the full
Knill--Laflamme equations decide realizability.
\end{definition}

The cutoff and support-dimension restriction retain, for every exact
\(((6,2,3))\) code and every nontrivial diagonal logical action on it, at
least one physical representative of that action. The \(66\) classes
enumerate the resulting filtered symmetry data, not all reduced physical
representatives.

We distinguish four objects: the full diagonal logical group, realizable
symmetry data, attainable weight enumerators, and code projectors up to
local unitaries.

\begin{definition}[Four classification levels]
\label{def:six-classification-levels}
For a code reconstructed from a chosen representative \((q,m,a)\) within
Eq.~\eqref{six:eq:qcutoff}, the four classification levels are as follows.
\begin{enumerate}[label=\textnormal{(\roman*)},leftmargin=2.4em]
\item The \emph{full logical-group level} concerns the entire cyclic group
      \(\mathcal G_{\mathrm{diag}}(P)\) of the reconstructed code in the
      chosen product basis.  The label \(m\) in the search data is the
      order of a chosen element; it need not equal the order of the full
      group.
\item The \emph{symmetry-data level} concerns the realizable
      classes \(r\) of
      Definition~\ref{def:six-sslp-orbits}, for which the full
      Knill--Laflamme completion set \(\mathcal K_r\) is nonempty.
\item For a realizable class \(r\), its
      \emph{weight-enumerator set} consists of all attainable vectors:
      \begin{equation}
        \mathcal W_r
        :=
        \{(A_0,\ldots,A_6):
        \text{an exact }((6,2,3))\text{ completion exists in }r\}.
        \label{eq:Wr-definition}
      \end{equation}
      The \(A_w\) are defined in Eq.~\eqref{eq:A-enumerator}; the associated
      \(B_w\) follow from the quantum MacWilliams transform.
\item At the \emph{local-unitary code level}, two projectors are identified
      under the local-unitary-permutation relation
      \eqref{eq:LU-equivalence}.  We also use the local-monomial relation
      defined there, which allows only bit flips, diagonal phase gates,
      and qubit permutations.
\end{enumerate}
\end{definition}

The same code can occur in several symmetry classes, either through
different diagonal actions or through different full fibers that contain
its actual support.  One realizable class can also contain a continuum
of enumerators, and different classes can have overlapping enumerator
sets.  The enumerators are invariant under local unitaries, but equality
of enumerators does not in general imply local-unitary equivalence.

As noted in Section~\ref{sec:qec-preliminaries}, the quantum Singleton
bound makes every six-qubit rank-two solution of the full
weight-at-most-two Knill--Laflamme equations an exact distance-three code.

\subsection{Classification results}
\label{six:sec:global-status}

The finite search now follows the three SSLP layers.  Subset-sum
enumeration, matching diagonal moments, and the support-dimension bound
leave \(66\) candidate classes after equivalent data are identified.
The full Knill--Laflamme equations exclude \(51\) classes, while exact
constructions realize the remaining \(15\).  The supplementary theorem
``The \(66\) SSLP moment orbits'' establishes the enumeration, and the
section ``Reproducible exact classification protocol'' gives the procedure.
The Supplementary Materials
call the candidates ``SSLP moment orbits'' and the realizable classes
``SSLP realization orbits.''  We use the orbit numbering in the
supplementary table listing all \(66\) rank-compatible moment-orbit
representatives.

We first determine the possible full diagonal logical groups and the
classes that realize them.  We then give thirteen exact enumerator sets
and the order-five uniqueness result.  Detailed proofs of the
classification results are collected under ``Proofs of the main-text
classification results'' in the Supplementary Materials.

\subsubsection{Result I: orders of the full diagonal logical groups}

Let \(M_3(6)\) denote the maximum projective logical order of a transversal
diagonal unitary preserving an exact \(((6,2,3))\) code, where the
maximum is taken over all such codes and transversal diagonal unitaries
in the chosen product basis.

\begin{theorem}[Six-qubit projective diagonal groups]
\label{six:thm:global-status}
For an exact \(((6,2,3))\) code whose projective diagonal logical group in
the chosen product basis is nontrivial,
\begin{equation}
  \boxed{
  \mathcal G_{\mathrm{diag}}(P)
  \cong C_m,
  \qquad
  m\in\{2,3,5\}.}
  \label{eq:six-intrinsic-groups}
\end{equation}
Equivalently, the possible projective element orders are
\(\{1,2,3,5\}\) including the identity, and
\begin{equation}
  \boxed{M_3(6)=5.}
  \label{eq:M36}
\end{equation}
\end{theorem}

The rank-two normal form makes the diagonal logical group cyclic.
Exact Knill--Laflamme exclusions and constructions leave precisely the
nontrivial element orders \(2,3,5\), which therefore also determine the
possible full groups.  See the proof of the six-qubit projective diagonal
group classification in the Supplementary Materials.

\subsubsection{Result II: realizable symmetry classes}

\begin{theorem}[The \(15\) realizable SSLP candidate classes]
\label{thm:six-15-sslp-realization-orbits}
Among the \(66\) SSLP candidate classes satisfying the support-dimension
bound, the full
weight-at-most-two Knill--Laflamme system admits an exact completion if and
only if
\begin{equation}
  \boxed{
  r\in
  \{1,2,3,4,7,8,9,10,11,16,21,22,23,36,45\}.}
  \label{six:eq:15-sslp-realization-orbits}
\end{equation}
There are exactly \(15\) realizable classes, distributed by the order
of the chosen logical action as follows:
\[
\begin{array}{c|rrr}
 m&2&3&5\\ \hline
 \#\text{ realizable classes}&11&3&1.
\end{array}
\]
Each listed class contains an exact \(((6,2,3))\) code.
\end{theorem}

The proof combines the exhaustive moment enumeration with exact
Knill--Laflamme no-go certificates for the excluded classes, including
zero-amplitude boundary cases.  Exact constructions in every listed
class give the converse.  See the proof of the classification of the
fifteen realizable symmetry classes in the Supplementary Materials.

Table~\ref{tab:six-realizable-classes} lists representatives and summarizes
the weight-enumerator sets.

\subsubsection{Result III: weight-enumerator sets}

A realizable symmetry class can contain codes with continuously varying
weight enumerators.  To describe this variation, define
\begin{align}
  \mathcal L(\alpha)
  &:=(1,\alpha,1-\alpha,0,
      11+4\alpha,16-\alpha,3-3\alpha),
  \label{eq:Lalpha}\\
  \mathcal M(t)
  &:=(1,0,t,1-t,12-t,15+t,3),
  \label{eq:Mt-main}\\
  \mathcal E_1^A
  &:=(1,1,0,0,15,15,0)=\mathcal L(1),
  \label{eq:E1A-main}\\
  \mathcal E_2^A
  &:=(1,0,1,0,11,16,3)=\mathcal L(0)=\mathcal M(1),
  \label{eq:E2A-main}\\
  \mathcal R_5^A
  &:=(1,0,21/25,0,291/25,384/25,79/25).
  \label{eq:R5A-main}
\end{align}

\begin{table}[!t]
\centering
\small
\setlength{\tabcolsep}{4pt}
\renewcommand{\arraystretch}{1.12}
\begin{tabular}{@{}rccll@{}}
\toprule
Class \(r\) & \(q\) & \(m\) & Representative \(a\)
  & Enumerator set \(\mathcal W_r\)\\
\midrule
1  & 2 & 2 & \((0,0,0,1,1,1)\) & Open\\
2  & 2 & 2 & \((0,0,1,1,1,1)\)
   & \(\{\mathcal L(\alpha):0\le\alpha\le1\}\)\\
3  & 2 & 2 & \((0,1,1,1,1,1)\)
   & \(\{\mathcal L(\alpha):0\le\alpha\le1\}\)\\
4  & 2 & 2 & \((1,1,1,1,1,1)\) & Open\\
\addlinespace
7  & 3 & 3 & \((0,1,1,1,2,2)\)
   & \(\{\mathcal L(\alpha):0\le\alpha\le1\}\)\\
8  & 3 & 3 & \((1,1,1,1,1,1)\)
   & \(\{\mathcal L(\alpha):\tfrac12\le\alpha\le1\}\)\\
9  & 3 & 3 & \((1,1,1,1,1,2)\) & \(\{\mathcal E_2^A\}\)\\
\addlinespace
10 & 4 & 2 & \((0,0,1,1,2,2)\) & \(\{\mathcal E_2^A\}\)\\
11 & 4 & 2 & \((0,0,1,2,2,2)\) & \(\{\mathcal E_1^A\}\)\\
16 & 4 & 2 & \((0,1,2,2,2,3)\) & \(\{\mathcal E_2^A\}\)\\
21 & 4 & 2 & \((1,1,1,2,2,3)\) & \(\{\mathcal E_2^A\}\)\\
22 & 4 & 2 & \((1,1,2,2,2,2)\) & \(\{\mathcal E_2^A\}\)\\
23 & 4 & 2 & \((1,2,2,2,2,2)\) & \(\{\mathcal E_1^A\}\)\\
\addlinespace
36 & 5 & 5 & \((1,1,1,2,2,4)\) & \(\{\mathcal R_5^A\}\)\\
45 & 6 & 2 & \((1,1,2,3,4,4)\)
   & \(\{\mathcal M(t):\tfrac23\le t\le1\}\)\\
\bottomrule
\end{tabular}
\caption{The \(15\) realizable classes.  Here \(q\) is the
physical order of the displayed representative, not necessarily the
minimum order of an implementation of the same logical action, and
\(m\) is the targeted projective logical
order.  Every row admits exact codes.  The two open enumerator sets
contain the exact families in
Eqs.~\eqref{eq:remaining-open-L-inclusion}--\eqref{eq:remaining-open-M-inclusion}.
Class~36 contains one code up to local unitaries and qubit permutations;
complete LU classification remains open for the other classes.}
\label{tab:six-realizable-classes}
\end{table}

\begin{theorem}[Exact weight-enumerator sets for thirteen classes]
\label{thm:six-WE-status}
For thirteen of the fifteen realizable classes in
Theorem~\ref{thm:six-15-sslp-realization-orbits}, the exact weight-enumerator
sets are as follows.
\begin{enumerate}[label=\textnormal{(\roman*)},leftmargin=2.4em]
\item Class~36 has
      \begin{equation}
        \mathcal W_{36}=\{\mathcal R_5^A\}.
        \label{eq:W36-main}
      \end{equation}
\item For projective logical order three,
      \begin{align}
        \mathcal W_7
        &=\{\mathcal L(\alpha):0\le\alpha\le1\},
        \label{eq:W7-main}\\
        \mathcal W_8
        &=\{\mathcal L(\alpha):\tfrac12\le\alpha\le1\},
        \label{eq:W8-main}\\
        \mathcal W_9
        &=\{\mathcal E_2^A\}.
        \label{eq:W9-main}
      \end{align}
\item For projective logical order two,
      \begin{align}
        \mathcal W_2=\mathcal W_3
        &=\{\mathcal L(\alpha):0\le\alpha\le1\},
        \label{eq:W23-main}\\
        \mathcal W_{10}=\mathcal W_{16}=\mathcal W_{21}=\mathcal W_{22}
        &=\{\mathcal E_2^A\},
        \label{eq:W-E2-main}\\
        \mathcal W_{11}=\mathcal W_{23}
        &=\{\mathcal E_1^A\},
        \label{eq:W-E1-main}\\
        \mathcal W_{45}
        &=\{\mathcal M(t):\tfrac23\le t\le1\}.
        \label{eq:W45-main}
      \end{align}
\end{enumerate}
The two sets for classes~1 and~4 remain unresolved.  The union of the
thirteen determined sets is
\begin{equation}
  \boxed{
  \bigcup_{r\in\{2,3,7,8,9,10,11,16,21,22,23,36,45\}}\mathcal W_r
  =
  \{\mathcal L(\alpha):0\le\alpha\le1\}
  \cup
  \{\mathcal M(t):\tfrac23\le t\le1\}
  \cup
  \{\mathcal R_5^A\}.}
  \label{eq:proved-WE-union}
\end{equation}
\end{theorem}

Each set is determined by a converse that bounds every exact completion
and a construction that attains every stated enumerator.  See
the proof of the weight-enumerator classification in the Supplementary
Materials.  For classes~2 and~3, the supplementary theorems ``Complete
orbit-2 enumerator image'' and ``Global orbit-3 enumerator theorem'' force
\(A_3=0\) for every realization,
including zero-amplitude and singular strata; the exact families then
attain the full \(\mathcal L\)-interval.

Classes~2, 3, and~7 share the full \(\mathcal L\)-interval, while class~8
realizes a proper subinterval.  Several other classes have the same
singleton enumerator set \(\{\mathcal E_2^A\}\).  These examples show how distinct
symmetry data can lead to the same enumerator set, even as one class
contains a continuous interval.

Classes~1 and~4 each realize the full \(\mathcal L\)-interval and the
\(\mathcal M\)-interval with \(2/3\le t\le1\).  The remaining question
is whether they realize any additional enumerators.  The constructions in
the supplementary theorem ``Exact family inclusions for orbits 1--4'' give

\begin{equation}
  \{\mathcal L(\alpha):0\le\alpha\le1\}
  \subseteq \mathcal W_r,
  \qquad r=1,4,
  \label{eq:remaining-open-L-inclusion}
\end{equation}
and
\begin{equation}
  \{\mathcal M(t):\tfrac23\le t\le1\}
  \subseteq \mathcal W_1\cap\mathcal W_4.
  \label{eq:remaining-open-M-inclusion}
\end{equation}
The numerical evidence is consistent with the sharper equality   
\begin{equation}
  \mathcal W_1\overset{?}{=}\mathcal W_4
  \overset{?}{=}
  \{\mathcal L(\alpha):0\le\alpha\le1\}
  \cup
  \{\mathcal M(t):\tfrac23\le t\le1\}.
  \label{eq:W14-conjecture}
\end{equation}
Numerical searches have found only these families.  An exact converse
excluding additional enumerators is still needed for each class; the
partial reductions are described under ``Unresolved weight-enumerator
images for orbits~1 and~4'' in the Supplementary
Materials.

We next compare these sets with the unrestricted code space.
Let
\begin{equation}
  \mathfrak W_{6,2,3}^A
  :=
  \left\{
    (A_0(P),\ldots,A_6(P)):
    P\text{ is the projector of an exact }((6,2,3))\text{ code}
  \right\}
  \label{eq:global-W-set}
\end{equation}
denote the \(A\)-enumerator set of the full \(((6,2,3))\) code space, with
no transversal-gate assumption.

The numerical search of Ref.~\cite{ZhangWuHuangZeng2025} imposed no
transversal-symmetry restriction and found two one-parameter families
and one isolated point.
In our parameterization, these observations motivate the conjecture
\begin{equation}
  \boxed{
  \mathfrak W_{6,2,3}^A
  \overset{?}{=}
  \{\mathcal L(\alpha):0\le\alpha\le1\}
  \ \cup\
  \{\mathcal M(t):\tfrac35\le t\le1\}
  \ \cup\
  \{\mathcal R_5^A\}.}
  \label{eq:global-WE-conjecture}
\end{equation}
The thirteen determined sets already realize the full
\(\mathcal L\)-interval, the segment \(2/3\le t\le1\) of the
\(\mathcal M\)-interval, and the isolated \(\mathcal R_5^A\) point.
Determining \(\mathcal W_1\) and \(\mathcal W_4\) would complete the
enumerator description of the nontrivial diagonal sector.  The global
conjecture also includes codes with trivial projective diagonal logical
group and proposes the wider range \(3/5\le t\le1\); those codes
require a separate argument.

\subsubsection{Result IV: the unique order-five code and its minimum length}

The order-five sector contains a single exact code up to local unitaries
and qubit permutations.  We display its codewords and transversal action
before stating the uniqueness and minimum-length results.  Split the
six physical qubits into the first four and the last two, and let
\(\omega=e^{2\pi i/3}\).  Define the unnormalized four-qubit vectors
\begin{align}
  \ket{S_j}
  &:=\sum_{\substack{x\in\{0,1\}^4\\\wt(x)=j}}\ket{x},
  \qquad j=1,3,
  \label{eq:six-C5-shells}\\
  \ket E
  &:=\ket{1100}+\ket{0011}
    +\omega(\ket{1010}+\ket{0101})
    +\omega^2(\ket{1001}+\ket{0110}),
  \label{eq:six-C5-edges}
\end{align}
and let \(\ket{E^*}\) be the coefficientwise complex conjugate of
\(\ket E\).  An orthonormal basis of the code is
\begin{align}
  \ket{\psi_0}
  &=\frac{\ket{0000}(\ket{00}+\ket{11})
           +\ket{S_3}\ket{10}}{\sqrt{10}}
    +\frac{\ket E\ket{01}}{\sqrt{15}},
  \label{eq:six-C5-code0}\\
  \ket{\psi_1}
  &=\frac{\ket{1111}(\ket{00}+\ket{11})
           -\ket{S_1}\ket{01}}{\sqrt{10}}
    -\frac{\ket{E^*}\ket{10}}{\sqrt{15}}.
  \label{eq:six-C5-code1}
\end{align}
This is a presentation of the code of
Ref.~\cite{ZhangWuHuangZeng2025}, locally equivalent to the class-36
representative in Table~\ref{tab:six-realizable-classes}.  In the displayed
basis, put
\begin{equation}
  a^\star=(1,1,1,1,2,3),\qquad \zeta=e^{2\pi i/5}.
  \label{eq:six-C5-gate-data}
\end{equation}
The two logical supports lie in \(S_0(a^\star;5)\) and
\(S_4(a^\star;5)\), respectively, so
\begin{equation}
  D(a^\star;5)\big|_{\mathcal Q}
  =\diag(1,\zeta^4),\qquad
  D(a^\star;5)^4\big|_{\mathcal Q}=\diag(1,\zeta).
  \label{eq:six-C5-logical-action}
\end{equation}
Both physical gates have order five.  The code satisfies the full
weight-at-most-two KL equations and has enumerator \(\mathcal R_5^A\)
in Eq.~\eqref{eq:R5A-main}; the KL verification and enumerator computation
are given under ``Phase rigidity and a canonical representative'' and
``The unique weight enumerator on orbit 36'' in the Supplementary Materials.
In particular, \(A_1=0\).

\begin{theorem}[The unique projective-\(C_5\) code]
\label{thm:unique-C5-LU}
Every exact \(((6,2,3))\) code admitting a projective diagonal logical
element of order five belongs to class~36.  All exact completions of
class~36 are equivalent under physical-qubit permutations, single-qubit
bit flips, and diagonal phase gates.  Consequently the projective-\(C_5\)
sector consists of
one local-unitary code class (with physical-qubit permutations allowed).
\end{theorem}

Projective logical order five forces a controlled physical representative
of order five.  Exact exclusions leave only class~36, whose amplitude and
phase equations give rigidity under bit flips, diagonal phase gates, and
qubit permutations.  See the proof of the order-five uniqueness theorem
in the Supplementary Materials.

\begin{corollary}[Minimum length for a projective order-five transversal gate]
\label{cor:six-C5-minimum-length}
If an exact \(((n,2,d))\) code with \(d\ge3\) supports a single-block
transversal unitary of projective logical order five, then \(n\ge6\).
Equality is attained by the code in
Eqs.~\eqref{eq:six-C5-code0}--\eqref{eq:six-C5-code1}, and every
six-qubit code with this property is equivalent to it under local
unitaries and qubit permutations.
\end{corollary}

The vanishing of \(A_1\) also explains the minimum length.  After locally
diagonalizing each transversal factor, a shorter code could be padded
with fixed ancillas to obtain a six-qubit code with the same order-five
logical action.  Each ancilla contributes one to \(A_1\), contradicting
the unique order-five code's \(A_1=0\).  The supplementary minimum-length
corollary for projective order five gives the full argument.

The same minimum applies to transversal non-Clifford logical gates in
the logical-basis-independent sense stated in the Introduction.  The
discussion following the supplementary minimum-length corollary for
projective order five combines the physical-order
bound with the classification of finite rotation groups: below six qubits, the full projective transversal logical group can always be conjugated into the Clifford group.

For projective logical orders two and three, classification at each fixed
enumerator value still requires determining the amplitude families, their
connected components, and their local-unitary equivalences.  Additional
symmetries at special parameter values and overlaps between different
symmetry classes are part of this finer problem.  Class~36 is the only
sector with a complete code-level classification.

\section{Discussion}
\label{sec:agentic-reproducibility}

The physical-order bounds make the symmetry data finite while allowing
continuous families of exact codes.  For exact qubit codes of distance at
least two, every nontrivial diagonal logical action in a fixed product
basis has a physical representative within a uniform finite range of
orders at fixed code parameters.
The SSLP normal form then places the code amplitudes in labeled subset-sum
fibers, where the full Knill--Laflamme equations determine the exact
completions.  Thus the finite outer enumeration and the continuous
completion problem can be treated separately without assuming stabilizer
or additive structure.

For exact \(((6,2,3))\) codes, this procedure reduces \(66\) candidate
symmetry classes to \(15\) realizable classes, with nontrivial projective
diagonal logical orders precisely \(\{2,3,5\}\).  We determine thirteen of
their fifteen weight-enumerator sets exactly.  Distinct classes can share
an enumerator set.  A classification up to local unitaries additionally
requires determining the exact code families within and across these
classes.  The order-five sector
admits the stronger conclusion: every exact completion belongs to one code
class up to local unitaries and qubit permutations.  This code also attains
the minimum length of six for an exact code encoding one logical qubit at
distance at least three and supporting a single-block transversal logical
action of projective order five.

Within the nontrivial diagonal sector, determining
\(\mathcal W_1\) and \(\mathcal W_4\) would complete the weight-enumerator
description.  A complete code-level classification of the order-two and
order-three classes must still determine the amplitude families at each
fixed enumerator value, their connected components, and local-unitary
equivalence among codes with the same enumerator inside and across classes.
These tasks also require identifying additional symmetries at special
parameter values.

The conjecture in Eq.~\eqref{eq:global-WE-conjecture} concerns the
unrestricted space of exact \(((6,2,3))\) codes.  A proof additionally
requires treating codes whose projective diagonal logical group is trivial.
Beyond this six-qubit problem, further directions include efficient
searches at larger code dimensions, stronger physical-order bounds at
larger distances, and transversal symmetries that cannot be diagonalized
in one common product basis.

The structural theorems were derived by the human authors and subsequently
formalized and verified with human guidance.  The supplementary enumeration,
case-by-case proofs, exact certificates, and Lean verification code were
produced by AI agents using TeXRA with GPT-5.6 and the Lean verification
workflow developed in Ref.~\cite{HeLuZeng2026}.  Humans audited the completed
supplementary results. The Supplementary Materials~\cite{SSLPClassificationSupplement}
and accompanying code, data, exact certificates, and Lean sources are available in the companion repository, together with instructions for reproducing the results.

\section*{Acknowledgments}
We thank Sirui Lu and Shawn Cui for helpful discussions.

\bibliographystyle{unsrt}
\bibliography{references}

\end{document}